\documentclass[aps,pra,groupedaddress,superscriptaddress,amsmath,amssymb,tightenlines,notitlepage,nofootinbib,twocolumn]{revtex4-2}
\pdfoutput=1

\usepackage[left=.73in,right=.73in,top=.7in,bottom=1in]{geometry}
\usepackage{graphicx}
\usepackage[dvipsnames,svgnames]{xcolor}
\usepackage{bm,bbm}%
\usepackage{braket}
\usepackage{amsthm,amsmath,amssymb,amsbsy}
\usepackage{enumerate}

\usepackage{moresize}
\usepackage{mathtools,mathrsfs,bbding}
\usepackage{microtype}
\usepackage{chngcntr}

\usepackage[T1]{fontenc}
\usepackage[utf8]{inputenc}

\definecolor{darkred}{rgb}{0.57,0,0.12}
\usepackage{hyperref}

\usepackage{cleveref}

\let\nc\newcommand

\usepackage{newpxtext,newpxmath}

\hypersetup{
    bookmarksnumbered=true, 
    breaklinks=true,
    unicode=false, 
    pdfstartview={FitH}, %
    pdftitle={Reversibility of quantum resources through probabilistic protocols}, 
    pdfnewwindow=true, 
    colorlinks=true, 
    allcolors=MidnightBlue!70!black!70!TealBlue
}

\nc{\note}[1]{{\color{blue!90!black} #1}}
\nc{\noteb}[1]{{\color{red!80!black}\textbf{#1}}}

\nc{\incoh}[1]{{\mathcal{C}^{(#1-1)}}}
\nc{\produc}[1]{{\mathcal{P}^{(#1)}}}

\nc{\cmin}{c_{\min}}

\nc{\Sch}{\mathrm{Sch}}
\nc{\CG}{C_{\text{GME}}}

\nc{\gram}[1]{G^{(#1)}}

\nc{\pen}{^{-1}}

\DeclareMathOperator{\Tr}{Tr}

\DeclareMathOperator{\NN}{NN}

\nc{\psic}{\psi^{c}}
\nc{\two}[1]{\underline{2^{d-#1}}}

\nc{\Hanc}{\mathcal{H}_{\text{anc}}}
\nc{\psianc}{\psi_{\text{anc}}}
\nc{\lampow}{\lambda^{1/d}}

\nc{\norm}[2]{\left\lVert#1\right\rVert_{\,#2}}
\nc{\proj}[1]{\ket{#1}\!\bra{#1}}
\nc{\pro}[1]{#1 #1^\dagger}
\nc{\lnorm}[2]{\left\lVert#1\right\rVert_{\ell_{#2}}}

\nc{\VV}{V^{\,\C}_\Sp}

\nc{\R}{\mathcal{R}}
\nc{\W}{\mathcal{W}}
\nc{\T}{\mathcal{T}}
\nc{\B}{\mathcal{B}}
\nc{\C}{\mathcal{C}}
\nc{\U}{\mathcal{U}}
\nc{\E}{\mathcal{E}}
\nc{\EE}{\mathscr{E}}
\nc{\K}{\mathcal{K}}
\nc{\V}{\mathcal{V}}
\nc{\X}{\mathcal{X}}
\nc{\F}{\mathcal{F}}
\nc{\G}{\mathcal{G}}
\nc{\D}{\mathcal{D}}
\nc{\Y}{\mathcal{Y}}

\nc{\M}{\mathcal{M}}
\nc{\N}{\mathcal{N}}
\nc{\I}{\mathcal{I}}
\nc{\Q}{\mathbb{Q}}
\nc{\RR}{\mathbb{R}}
\nc{\CC}{\mathbb{C}}

\nc{\HH}{\mathbb{H}}
\nc{\MM}{\mathbb{M}}
\nc{\DD}{\mathbb{D}}
\renewcommand{\NN}{\mathbb{N}}

\nc{\J}{\mathcal{J}}

\nc{\FF}{\mathbb{F}}
\nc{\OO}{\mathbb{O}}

\newcommand\floor[1]{\left\lfloor#1\right\rfloor}
\newcommand\ceil[1]{\left\lceil#1\right\rceil}

\let\succc\succ

\newcommand{\PPT}{\text{\rm PPT}}

\nc{\mleq}{\preceq}
\nc{\mgeq}{\succeq}
\nc{\mle}{\prec}
\nc{\mge}{\succc}

\renewcommand{\succ}{{\mathrm{succ}}}

\nc{\CPTP}{{\mathrm{CPTP}}}
\nc{\CP}{{\mathrm{CP}}}
\nc{\CPTNI}{{\mathrm{CPTNI}}}

\nc{\ext}[1]{\operatorname{ext}\left(#1\right)}

\nc{\<}{\left\langle}

\renewcommand{\bar}{\;\rule{0pt}{9.5pt}\right|\;}

\nc{\lset}{\left\{\left.}
\nc{\rset}{\right\}}

\nc{\ve}{\varepsilon}
\nc{\cbraket}[1]{\left|\braket{#1}\right|}
\nc{\id}{\mathbbm{1}}
\nc{\idc}{\mathrm{id}}

\nc{\mnorm}[1]{\norm{#1}{[m]}}
\nc{\knorm}[1]{\norm{#1}{(k)}}

\theoremstyle{plain}

\newtheorem{theorem}{Theorem}

\newtheorem{proposition}[theorem]{Proposition}

\newtheorem{lemma}[theorem]{Lemma}

\theoremstyle{definition}

\newtheorem*{remark}{Remark}

\let\oldproofname\proofname
\renewcommand{\proofname}{\rm\bf{\oldproofname}}

\nc{\lsetr}{\left\{\,}
\nc{\rsetr}{\right.\right\}}
\nc{\barr}{\;\rule{0pt}{9.5pt}\left|\;}

\nc{\prob}{\mathrm{prob}}
\nc{\wt}{\widetilde}

\nc{\bb}{{\bullet\bullet}}
\nc{\phim}{{\phi_{\star}}}

\let\epsilon\varepsilon

\nc{\PPTP}{{\rm PPTP}}

\nc{\wtD}{\wt{D}}

\let\textleq\relax
\let\textgeq\relax
\let\texteq\relax

\newcommand{\texteq}[1]{\stackrel{\mathclap{\scriptsize \mbox{#1}}}{=}}
\newcommand{\textleq}[1]{\stackrel{\mathclap{\scriptsize \mbox{#1}}}{\leq}}

\newcommand{\textgeq}[1]{\stackrel{\mathclap{\scriptsize \mbox{#1}}}{\geq}}

\makeatletter 
\renewcommand\onecolumngrid{%
\do@columngrid{one}{\@ne}%
\def\set@footnotewidth{\onecolumngrid}%
\def\footnoterule{\kern-6pt\hrule width 1.5in\kern6pt}%
}
\makeatother 

\usepackage[most,breakable]{tcolorbox}
\renewenvironment{boxed}[1][white]%
  {\expandafter\ifstrequal\expandafter{#1}{filled}{\begin{tcolorbox}[colback=gray!3,colframe=gray!20,breakable,enhanced,left=5.75pt,right=5.75pt,grow sidewards by=10pt]}{\begin{tcolorbox}[colback=white,colframe=gray!15,breakable,enhanced,left=5.75pt,right=5.75pt,grow sidewards by=10pt]}}%
  {\end{tcolorbox}}

\usepackage{stackengine}
\stackMath
\newcommand\xxrightarrow[2][]{\mathrel{%
  \setbox2=\hbox{\stackon{\scriptstyle#1}{\scriptstyle#2}}%
  \stackunder[0.4pt]{%
    \xrightarrow{\makebox[\dimexpr\wd2\relax]{$\scriptstyle#2$}}%
  }{%
   \scriptstyle#1%
  }%
}}

\newcommand{\tends}{\xxrightarrow[\! n\rightarrow \infty\!]{}}

\newcommand{\deff}[1]{\textbf{\emph{#1}}}

\newcommand{\toO}{\xxrightarrow[\vphantom{a}\smash{{\raisebox{0pt}{\ssmall{$\OO$}}}}]{}}
\newcommand{\toARNG}{\xxrightarrow[\vphantom{a}\smash{{\raisebox{0pt}{\ssmall{${\rm ARNG}$}}}}]{}}
\newcommand{\toARNGs}{\xxrightarrow[\vphantom{a}\smash{{\raisebox{0pt}{\ssmall{${\rm ARNG},s$}}}}]{}}
\newcommand{\toANEN}{\xxrightarrow[\vphantom{a}\smash{{\raisebox{0pt}{\ssmall{${{\rm ANE},N}$}}}}]{}}
\newcommand{\toANEs}{\xxrightarrow[\vphantom{a}\smash{{\raisebox{0pt}{\ssmall{${\rm ANE},s$}}}}]{}}
\newcommand{\toNE}{\xxrightarrow[\vphantom{a}\smash{{\raisebox{0pt}{\ssmall{${\rm NE}$}}}}]{}}

\nc{\RNG}{{\rm RNG}}
\nc{\NE}{{\rm NE}}
\nc{\ANE}{{\rm ANE}}
\nc{\ANEN}{{\rm ANE},N}
\nc{\ANEs}{{\rm ANE},s}
\nc{\ARNG}{{\rm ARNG}}
\nc{\RNGdelta}{{{\rm RNG}_\delta}}
\nc{\RNGdeltan}{{{\rm RNG}_{\delta_n}}}
\newcommand{\SEP}{\text{\rm SEP}}

\makeatletter
\patchcmd{\@ssect@ltx}
    {\addcontentsline{toc}{#1}{\protect\numberline{}#8}}
    {}
    {}
    {}
\def\l@subsection#1#2{}
\def\l@subsubsection#1#2{}
\makeatother

\newcommand{\para}[1]{\subsection*{#1}}
\newcommand{\parait}[1]{\smallskip\textit{#1}.\,---}

\nc{\ba}{\begin{equation}\begin{aligned}}
\nc{\ea}{\end{aligned}\end{equation}}
\let\bb\ba
\let\ee\ea
\let\e\ve
\let\ketbra\proj

\begin{document}

\title{Reversibility of quantum resources through probabilistic protocols}

\author{Bartosz Regula}
\email{bartosz.regula@gmail.com}
\affiliation{Mathematical Quantum Information RIKEN Hakubi Research Team, RIKEN Cluster for Pioneering Research (CPR) and RIKEN Center for Quantum Computing (RQC), Wako, Saitama 351-0198, Japan}

\author{Ludovico Lami}
\email{ludovico.lami@gmail.com}
\affiliation{QuSoft, Science Park 123, 1098 XG Amsterdam, the Netherlands}
\affiliation{Korteweg--de Vries Institute for Mathematics, University of Amsterdam, Science Park 105-107, 1098 XG Amsterdam, the Netherlands}
\affiliation{Institute for Theoretical Physics, University of Amsterdam, Science Park 904, 1098 XH Amsterdam, the Netherlands}

\begin{abstract}%
Among the most fundamental questions in the manipulation of quantum resources such as entanglement is the possibility of reversibly transforming all resource states. The key consequence of this would be the identification of a unique entropic resource measure that exactly quantifies the limits of achievable transformation rates. Remarkably, previous results claimed that such asymptotic reversibility holds true in very general settings; however, recently those findings have been found to be incomplete, casting doubt on the conjecture. Here we show that it is indeed possible to reversibly interconvert all states in general quantum resource theories, as long as one allows protocols that may only succeed probabilistically. Although such transformations have some chance of failure, we show that their success probability can be ensured to be bounded away from zero, even in the asymptotic limit of infinitely many manipulated copies. As in previously conjectured approaches, the achievability here is realised through operations that are asymptotically resource non-generating, and we show that this choice is optimal: smaller sets of transformations cannot lead to reversibility. Our methods are based on connecting the transformation rates under probabilistic protocols with strong converse rates for deterministic transformations, which we strengthen into an exact equivalence in the case of entanglement distillation. 
\end{abstract}

\maketitle

\stepcounter{part}


How to measure and compare quantum resources? As evidenced by the plethora of commonly used quantifiers of resources such as entanglement~\cite{horodecki_2001,chitambar_2019}, this seemingly basic question has many possible answers, and it may appear as though there is no unambiguous way to resolve it. However, it is important to keep in mind that, besides simply assigning some numerical value to a given quantum state, one often wishes to compare quantum resources operationally. Is it then possible for a resource quantifier to indicate exactly how difficult it is to convert one resource state into another? Following this pathway is reminiscent of the operational approach used to study thermodynamics, where indeed a unique resource measure --- the entropy --- emerges naturally from basic axioms~\cite{giles_1964,lieb_1999}. A phenomenon that is intimately connected with the existence of such a measure is reversibility: two comparable states of equal entropy can always be connected by a reversible adiabatic transformation~\cite{giles_1964,lieb_1999}.

Reversibility was observed also in the asymptotic manipulation of quantum entanglement of pure states~\cite{bennett_1996-1}, prompting several conjectures about the connections between entanglement and thermodynamics, and in particular about the existence of a unique operational measure of entanglement that would mirror the role of entropy~\cite{popescu_1997-2,vedral_2002,horodecki_2002,brandao_2008-1,brandao_2010}. 
In the asymptotic setting, reversibility is typically understood in terms of asymptotic transformation rates $r(\rho \to \omega)$: given many copies of a state $\rho$, how many copies of another state $\omega$ can we obtain per each copy of $\rho$? The question of resource reversibility then asks whether $r(\rho \to \omega) = r(\omega \to \rho)^{-1}$, meaning that exactly as many copies of $\omega$ can be obtained in the transformation as are needed to transform them back into $\rho$. Although entanglement of noisy states may exhibit irreversibility in many contexts~\cite{horodecki_1998,vidal_2001,wang_2017-1,lami_2023}, hopes persisted that an operational approach allowing for universal reversibility could be constructed, leading to the identification of a unique asymptotic measure of entanglement~\cite{OpenProblemArxiv}.

A remarkable axiomatic framework emerged, first for entanglement~\cite{brandao_2008-1,brandao_2010} and later for more general quantum resources~\cite{brandao_2015}, which claimed that reversibility can indeed always be achieved under suitable assumptions. Such a striking property would not only establish a unique entropic measure of quantum resources, but also connect the broad variety of different resources in a common operational formalism. However, issues have transpired in parts of the proof of these results~\cite{berta_2022,berta_2023-1}, putting this general reversibility into question. As of now, there is no known framework that can establish the reversibility of general quantum resource theories --- and in particular quantum entanglement --- even under weaker assumptions. What is more, recent results demonstrated an exceptionally strong type of irreversibility of entanglement~\cite{lami_2023}, casting doubt on the very possibility of recovering reversible manipulation whatsoever. 

In this work, we resolve the question by constructing the first complete reversible framework for general quantum resources, including entanglement.
Our setting closely resembles the original assumptions of the reversibility conjectures~\cite{brandao_2008-1,brandao_2010,brandao_2015}, with only one change: we allow probabilistic conversion protocols. That is, we study transformations which allow for some probability of failure, and we demonstrate that in this setting the conversion rates are exactly given by the entropic resource measure known as the regularised relative entropy, identifying it as the unique operational resource quantifier.

The use of probabilistic protocols is a distinguishing feature of our approach, necessitating a careful consideration of how exactly to quantify the asymptotic rates of such transformations. We employ a definition which focuses on the number of copies of states that are undergoing the transformation, discounting the success probability of the protocols.
Although seemingly more permissive than some previously used definitions, we explicitly ensure that the protocols that we study are not unphysically difficult to realise: we only allow transformations whose probability of failure does not become prohibitively large, in the sense that there always remains a constant non-zero chance of successful conversion, even when manipulating an unbounded number of quantum states.
Such rates are well behaved and closely connected to conventional asymptotic transformation rates  studied in quantum information.

We stress that, although conceptually similar, our approach follows an alternative pathway that does not exactly recover the reversibility conjectured in Refs.~\cite{brandao_2008-1,brandao_2010,brandao_2015}, as the latter relied only on strictly deterministic transformation rates. However, in light of the similarities and relations that we establish between probabilistic and deterministic rates,
we consider our results to be strong supporting evidence in favour of reversibility being an achievable phenomenon in general quantum resource theories.

On the technical side, the way we avoid issues associated with the generalised quantum Stein's lemma~\cite{brandao_2010-1} that undermined the original reversibility claims, is to use only the \emph{strong converse} part of the lemma, which is still valid~\cite{berta_2022}. Strong converse rates are typically understood as general no-go limitations on resource transformations, but here we turn them into achievable rates precisely by employing probabilistic protocols.
For the special case of entanglement distillation, we show that these two concepts --- strong converse rates in deterministic transformations on one side, and probabilistic conversion rates on the other --- are exactly equivalent, which holds true also 
in the most practically relevant settings of entanglement manipulation such as under local operations and classical communication (LOCC).


\section*{Results}

\para{Resource transformation rates} Quantum resource theories represent various settings of restricted quantum information processing~\cite{chitambar_2019}. Let us denote by $\OO$ the set of operations which are freely allowed within the physical setting of the given resource theory. Our discussion of reversibility will require a specific choice of $\OO$, but for now it may be understood as a general set of permitted operations.

The \emph{deterministic transformation rate} $r_{p=1}(\rho \to \omega)$ is defined as the supremum of real numbers $r$ such that $n$ copies of the state $\rho$ can be 
converted into $\floor{rn}$ copies of the target state $\omega$ using the free operations $\OO$. The conversion here is assumed to be deterministic, i.e.\ all transformations are realised by completely positive and trace-preserving maps (quantum channels). However, the process is only required to be approximate, in the sense that some error $\ve_n$ is allowed in the transformation, as long as it vanishes in the limit as $n\to\infty$.

In many practical contexts, one may be willing to relax the assumption that the error must tend to zero --- it may, for instance, be appealing to tolerate some manageable error in the transformation if it could lead to increased transformation rates. The ultimate upper bound that constrains the improvements that can be gained through such trade-offs is represented by the \emph{strong converse rate} $r^\dagger_{p=1}(\rho \to \omega)$. It is defined as the least rate $r$ such that, if we attempt the conversion $\rho^{\otimes n} \to \omega^{\otimes \floor{r' n}}$ at any larger rate $r' > r$, then 
even approximate transformations with large error become impossible.

Another common way to increase the capabilities in resource manipulation is to allow for probabilistic transformations~\cite{bennett_1996,horodecki_1999-1,vidal_1999-1,lo_2001,regula_2022,kondra_2022,regula_2023}. Probabilistic protocols in quantum information theory are represented by a collection of completely positive, but not necessarily trace-preserving maps $\{ \E^{(i)} \}_i$, such that the total transformation $\sum_i \E^{(i)}$ preserves trace. We say that $\rho$ can be converted to $\omega$ if there exists a free probabilistic operation $\E^{(i)} \in \OO$ such that $\frac{\E^{(i)}(\rho)}{\Tr \E^{(i)}(\rho)} = \omega$; the probability of this transformation is $p = \Tr \E^{(i)}(\rho)$. 
The question then is how to exactly define the asymptotic rate of such protocols.

Consider a sequence of probabilistic operations $(\E_n)_n$ such that each $\E_n \in \OO$ converts $\rho^{\otimes n}$ to a state which is $\varepsilon_n$-close to the target state $\omega^{\otimes \floor{rn}}$ with the error vanishing asymptotically.
We will write $p_n = \Tr \E_n(\rho^{\otimes n})$ for the probability of successful conversion. 
One way to quantify the rate of such a protocol is to count the average number of copies of quantum states needed to realise the transformation, which means that our rate would be given by $p_n r$ rather than just $r$, since the protocol needs to be repeated $1/p_n$ times on average to ensure success.
But one may argue that there is an issue with such a definition: is it fair to say that manipulating $n$ copies of a quantum state 
about $1/p_n$ times is as difficult as manipulating the larger number of $n / p_n$ copies all at once? This definition of a rate would make it seem so, since it counts the total number of copies needed in the protocol, and disregards the question of how many of those copies need to be coherently manipulated together. If the rate is supposed to quantify the difficulty in performing a state transformation, then this may not be an accurate assessment, considering that it is the manipulation of quantum states, rather than their generation, that is typically the bottleneck in practical quantum information processing.

An alternative way is then to simply say that, if $\rho^{\otimes n}$ is approximately converted to $\omega^{\otimes \floor{rn}}$ --- even probabilistically --- then the rate is $r$.
This definition focuses on the number of copies of states that are being transformed at once, and it does not count the probability $p_n = \Tr \E_n(\rho^{\otimes n})$ itself as part of the rate. 
However, there is again a potential issue with this approach, as leaving the probability of the transformation unconstrained effectively allows for conditioning on exponentially unlikely events, which then makes possible transformations that are conventionally known to be unachievable~\cite{regula_2023}. Such a phenomenon happens when the overall probability of success $p_n$ becomes vanishingly small. 
In order to exclude such unphysical protocols which cannot be implemented in practice, it thus becomes necessary to carefully constrain the success probability.

Our approach will then aim to find a middle ground: we will quantify rates in a way that only counts the number of transformed states, but we will explicitly forbid the possibility of the success probability becoming unphysically small. Specifically, to ensure that any considered protocol remains practically realisable, we will assume that the conversion probability is bounded away from zero. 
We thus consider \emph{probabilistic transformation rates with non-vanishing probability}, $r_{p>0}(\rho \to \omega)$, defined as

\begin{equation}\begin{aligned}
    r_{p>0}(\rho \to \omega) \coloneqq &\sup_{(\E_n)_n} \Bigg\{ \,r \;\Bigg|\; \E_n \in \OO,\\
    & \lim_{n\to\infty} F\!\left( \frac{\E_n(\rho^{\otimes n})}{\Tr \E_n(\rho^{\otimes n})} ,\, \omega^{\otimes \floor{rn}} \right) = 1,\\
     & \liminf_{n\to\infty}\, \Tr \E_n(\rho^{\otimes n}) > 0 \Bigg\},
\end{aligned}\end{equation}
where $F$ denotes fidelity.
We stress that this imposes a strong restriction on the allowed protocols, as the overall probability of success must remain larger than a positive constant, even in the asymptotic limit where the number of transformed copies grows to infinity. In short, the number of copies of $\rho$ that need to be \emph{manipulated} at once to obtain $m$ copies of $\omega$ is asymptotically $n = m / r_{p>0}(\rho \to \omega)$, while the total number of copies of $\omega$ that need to be \emph{generated} is that number times a constant overhead.

To further motivate our choice of definition of a probabilistic rate, let us compare it with the two deterministic rates introduced in this section. 
The fact that the deterministic transformation rate $r_{p=1}$ is the smallest of the three types is clear from the definition. However, there is no obvious relation between the strong converse rate and the probabilistic one. We can nevertheless show that the rates actually form a hierarchy:
\begin{equation}\begin{aligned}\label{eq:rate_hierarchy}
    r_{p=1}(\rho \to \omega) \,\leq\, r_{p > 0}(\rho \to \omega) \,\leq\, r^\dagger_{p=1}(\rho \to \omega).
\end{aligned}\end{equation}
This demonstrates in particular that the probabilistic rate $r_{p>0}$ is well behaved, as it does not exceed conventional limitations imposed by strong converse rates. It also naturally fits into the information-theoretic framework for asymptotic transformations and may even provide a tighter restriction on deterministic transformation rates than those coming from strong converse bounds.


\para{Free operations and reversibility} The asymptotic transformation rates depend heavily on the choice of the free operations $\OO$. Typically, practically relevant choices of free operations are subsets of \emph{resource--non-generating (RNG) operations}, defined as those maps $\E$ (possibly probabilistic ones) that satisfy $\frac{\E(\sigma)}{\Tr \E(\sigma)} \in \FF$ for all $\sigma \in \FF$. 
Here, $\FF$ stands for the set of free (resourceless) states of the given theory. The definition of RNG operations then means that these maps are not allowed to generate any resources for free, which is a very basic and undemanding assumption to make.

The framework of~\cite{brandao_2008-1,brandao_2010,brandao_2015} studied the manipulation of quantum resources under transformations which slightly relax the above constraint, imposing instead that small amounts of resources may be generated, as long as they vanish asymptotically. Specifically, let us consider the resource measure known as \emph{generalised (global) robustness} $R_\FF^g$, defined as~\cite{vidal_1999}
\begin{equation}\begin{aligned}\label{eq:gen_rob_def}
    R^g_\FF (\rho) \coloneqq \inf \lsetr \lambda \in \RR_+ \barr \frac{\rho + \lambda \omega}{1+\lambda} \in \FF,\; \omega \in \DD \rsetr,
\end{aligned}\end{equation}
where $\DD$ denotes the set of all states. The $\delta$-approximately resource--non-generating operations $\OO_{\RNGdelta}$ are then all maps $\E$ such that
\begin{equation}\begin{aligned}\label{eq:rng_def_delta}
    R^g_\FF \left( \frac{\E(\sigma)}{\Tr \E(\sigma)} \right) \leq \delta \quad \forall \sigma \in \FF.
\end{aligned}\end{equation}
Finally, the transformation rates under \emph{asymptotically resource--non-generating maps} $\OO_{\rm ARNG}$, whether deterministic or probabilistic, are defined as those where each transformation $\rho^{\otimes n} \to \omega^{\otimes \floor{rn}}$ is realised by a $\delta_n$-approximately RNG operation, with $\delta_n \to 0$ in the limit as $n\to\infty$. We will denote deterministic rates under such operations as $r_{p=1}(\rho \toARNG \omega)$, and analogously for the probabilistic rates $r_{p>0}$.

The main reason to study asymptotically RNG operations is their conjectured reversibility~\cite{brandao_2008-1,brandao_2010,brandao_2015}. 
Specifically, the claim is that the deterministic rates always equal
\begin{equation}\begin{aligned}\label{eq:BP_conjecture}
    r_{p=1}(\rho \toARNG \omega) \stackrel{?}{=} \frac{D^\infty_\FF(\rho)}{D^\infty_\FF(\omega)},
\end{aligned}\end{equation}
where $D^\infty_\FF$ denotes the \emph{regularised relative entropy of a resource},
\begin{equation}\begin{aligned}
    D^\infty_\FF(\rho) \coloneqq \lim_{n\to\infty} \frac1n 
    \,\inf_{\sigma \in \FF}\, D(\rho^{\otimes n} \| \sigma_n)
\end{aligned}\end{equation}
with $D(\rho\|\sigma) = \Tr \rho (\log\rho - \log\sigma)$ being the quantum relative entropy. This would precisely identify $D^\infty_\FF$ as the unique resource measure in the asymptotic setting. However, this conjecture relied crucially on the generalised quantum Stein's lemma~\cite{brandao_2010-1}, in whose proof a gap was recently discovered~\cite{berta_2022}. Hence, the statement in Eq.~\eqref{eq:BP_conjecture} is not known to be true~\cite{berta_2022}.

One may also wonder whether there are other possible candidates for operations that could lead to reversibility. This is especially relevant since the asymptotically resource non-generating maps $\OO_\ARNG$ are defined in an axiomatic way, and it may be appealing to study smaller classes of transformations constructed through more practically-minded considerations. 
However, such a possibility has been ruled out: in the context of deterministic transformations, essentially all sets of operations smaller than $\OO_\ARNG$ have been shown to lead to an irreversible theory of entanglement~\cite{lami_2023}. Importantly, the choice of the measure $R^g_\FF$ in the definition of $\OO_\ARNG$ is crucial, and even a small change of the resource quantifier can preclude reversibility. What this means is that any reversible theory of entanglement must actually generate exponentially large amounts of entanglement according to certain measures~\cite{lami_2023}.
But even such entanglement generation is not enough on its own: if one requires the considered transformations to also be `dually' resource non-generating (i.e., in the Heisenberg picture), then reversibility is impossible, even if one permits generating entanglement akin to $\OO_\ARNG$~\cite{DNE-distillable}. 
On the other hand, choosing more permissive types of operations, such as ones where the generated resources are quantified with the relative entropy $D^\infty_\FF$, may be too lax of a constraint, as such a theory trivialises by allowing for the distillation of unbounded amounts of entanglement~\cite{brandao_2010}. 
Altogether, this provides a strong motivation to study reversibility precisely under the class $\OO_\ARNG$, as it constitutes a `Goldilocks' 
set of operations that may allow for reversible manipulation while maintaining reasonable restrictions on the allowed transformations.


\para{Probabilistic reversibility} The conjectured resource reversibility (Eq.~\eqref{eq:BP_conjecture}) was formulated in a remarkably general manner. The original claim was meant to apply not only to entanglement, but also to more general quantum resources, as long as the set $\FF$ satisfies a number of mild assumptions --- notably, it must be convex, and it must be such that the tensor product of any two free states remains free, as does their partial trace~\cite{brandao_2010-1,brandao_2015}. These are weak assumptions obeyed by the vast majority of theories of practical interest.

Our main result is a general probabilistic reversibility of quantum resources under the exact same assumptions.

\begin{theorem}\label{thm:reversibility}
For all quantum states $\rho$ and $\omega$, the transformation rate with non-vanishing probability of success under asymptotically resource--non-generating operations satisfies
\begin{equation}\begin{aligned}
    r_{p>0}(\rho \toARNG \omega) = \frac{D^\infty_\FF(\rho)}{D^\infty_\FF(\omega)}.
\end{aligned}\end{equation}
This implies in particular a general reversibility of state transformations: $r_{p>0}(\rho \toARNG \omega) =  r_{p>0}(\omega \toARNG \rho)^{-1}$ for all pairs of states.
\end{theorem}
Both the converse and the achievability parts of this result make use of the asymptotic equipartition property for the generalised robustness, which was shown by Brand\~ao and Plenio~\cite[Proposition~II.1]{brandao_2010-1} and independently by Datta~\cite[Theorem~1]{datta_2009-2}. This property says that, under a suitable `smoothing', the generalised robustness $R^g_\FF$ converges asymptotically to the regularised relative entropy of the resource:
\begin{equation}\begin{aligned}\label{eq:aep}
    \lim_{\ve\to0}\limsup_{n\to\infty} \frac1n \min_{\frac12 \norm{\omega_n - \omega^{\otimes \floor{rn}}}{1} \leq \ve}\!\log\! \left( 1+R^g_\FF(\omega_n)\right) = r \, D^\infty_\FF(\omega).
\end{aligned}\end{equation}
Importantly, this finding directly leads to the strong converse of the generalised quantum Stein's lemma~\cite[Corollary~III.3]{brandao_2010-1}, but it does not appear to be enough to deduce the main achievability part of the lemma~\cite{berta_2022}, which underlies the previous reversibility conjectures. Our main contribution here is to show that the strong converse part is sufficient to show the reversibility of quantum resources, as long as probabilistic protocols are allowed.
\begin{proof}[Proof sketch of Theorem~\ref{thm:reversibility}]
The converse direction relies on the strong monotonicity properties of the generalised robustness $R^g_\FF$ as well as the aforementioned asymptotic equipartition property~\eqref{eq:aep}. This follows a related approach that was recently used to study postselected probabilistic transformation rates~\cite{regula_2023}, and here we extend it to asymptotically resource--non-generating transformations ARNG. A point of note is that standard techniques for upper bounding transformations rates, based on the asymptotic continuity of the relative entropy~\cite{horodecki_2001,horodecki_2013-3,winter_2016-1}, do not seem to be sufficient to establish a converse bound on probabilistic rates~\cite[Appendix~H]{regula_2023}. Our approach requires the use of a different toolset that explicitly makes use of the features of $R^g_\FF$.

For the achievability part of the theorem, 
we use the exact calculation of the strong converse exponent in the generalised quantum Stein's lemma~\cite{brandao_2010-1}. The lemma is concerned with the distinguishability of many copies of a quantum state $\rho^{\otimes n}$ against all states in the set of the free states $\FF$. 
The result of~\cite{brandao_2010-1} then says that, for every resource theory, there exists a sequence of measurement operators $(A_n)_n$ such that $\Tr \left( A_n \rho^{\otimes n} \right) \geq 1-\delta_n$ and
\begin{equation}\begin{aligned}\label{eq:stein_sc}
- \frac1n \log \sup_{\sigma \in \FF} \Tr( A_n  \sigma) \tends D^\infty_\FF(\rho).
\end{aligned}\end{equation}
Here, $\delta_n$ denotes the probability of incorrectly guessing that $\rho^{\otimes n}$ is a free state, while the quantity in Eq.~\eqref{eq:stein_sc} characterises the opposite error of incorrectly guessing that a free state is $\rho^{\otimes n}$.
The issue here is that the proof of Ref.~\cite{brandao_2010-1}, and hence also Eq.~\eqref{eq:stein_sc}, 
is only valid in the strong converse regime: this means that the error $\delta_n$ is not guaranteed to vanish, but it may actually tend to a constant arbitrarily close to 1. This prevents a direct application of previous deterministic results~\cite{brandao_2015}. 

What we do instead is 
define probabilistic operations of the form
\begin{equation}\begin{aligned}
    \E_n(\tau) \coloneqq \Tr (A_n \tau)\, \omega_n + \mu_n  \Tr [ (\id - A_n) \tau ]\, \pi_n,
\end{aligned}\end{equation}
where: $\omega_n$ are states appearing in~\eqref{eq:aep} which are ${\ve_n}$-close to the target states $\omega^{\otimes \floor{n \, D^\infty_\FF(\rho)/D^\infty_\FF(\omega)}}$, $\pi_n$ are some suitably chosen states, and $\mu_n \in [0,1]$ are parameters to be fixed. 

The basic idea is then that by decreasing $\mu_n$, we can make the output of this operation closer to $\omega_n$, even when $\Tr (A_n \rho^{\otimes n}) \not\to 1$. However, one cannot just decrease $\mu_n$ arbitrarily, as the maps $\E_n$ must be ensured to be free operations. Our crucial finding is that $\mu_n$ can always be chosen so that $\mu_n \tends 0$ while the operations $\E_n$ generate asymptotically vanishing amounts of resources and the overall probability of success does not vanish. 
This means precisely that the sequence $(\E_n)_n$ is an ARNG protocol that realises the desired conversion.
\end{proof}

The complete proof of Theorem~\ref{thm:reversibility} can be found in Section~\ref{sec:app_rev} of the~\hyperlink{supp}{Supplementary Information}.


\para{Optimality of ARNG transformations} To provide a stronger motivation for the choice of asymptotically resource--non-generating operations in the study of reversible transformations, we can show that this set of operations is essentially the smallest possible: more restrictive types of operations cannot lead to reversibility in general, even under probabilistic transformations. 

One natural way to constrain the allowed resource transformations is to forbid resource generation --- that is, instead of asymptotically resource--non-generating maps, consider strictly resource-non--generating ones.
An even more fine-grained restriction can be obtained by allowing for approximate resource generation, but choosing a more restrictive resource measure with which to quantify the generated resources.
To be precise, let us consider a modified notion of $\delta$-approximately RNG transformations that we will call $\OO_{\RNG,\delta,s}$. They are defined in exactly the same manner as ARNG rates, but instead of constraining the generalised robustness $R^g_\FF$ as in Eq.~\eqref{eq:rng_def_delta}, we impose that
\begin{equation}\begin{aligned}\label{eq:rngs}
    R^s_\FF \left( \frac{\E(\sigma)}{\Tr \E(\sigma)} \right) \leq \delta \quad \forall \sigma \in \FF,
\end{aligned}\end{equation}
where $R^s_\FF$ denotes the \emph{standard robustness}~\cite{vidal_1999}
\begin{equation}\begin{aligned}\label{eq:st_rob_def}
    R^s_\FF (\rho) \coloneqq \inf \lsetr \lambda \in \RR_+ \barr \frac{\rho + \lambda \sigma}{1+\lambda} \in \FF,\; \sigma \in \FF \rsetr.
\end{aligned}\end{equation}
This measure is very similar to the generalised robustness of Eq.~\eqref{eq:gen_rob_def}, but the state $\sigma$ is now required to be free; because of this, it holds that $R^s_\FF(\rho) \geq R^g_\FF(\rho)$ in general, making the constraint in~\eqref{eq:rngs} potentially more restrictive than before. 
We stress that strictly resource--non-generating transformations $\OO_{\RNG}$ are a subset of $\OO_{\RNG,\delta,s}$, so all irreversibility results shown for the latter apply also to the former.

We can use this to define modified transformation rates $r\big(\rho \toARNGs \omega\big)$ as those realised under $\OO_{\RNG,\delta_n,s}$ transformations with $\delta_n \to 0$ as $n \to \infty$.

By extending the methods that we used previously in the study of deterministic irreversibility~\cite{lami_2023}, we can show the following.
\begin{theorem}\label{thm:irrev_main}
Even in the probabilistic setting, general reversibility is not possible under operations that do not generate any resources or ones that only generate asymptotically vanishing amounts of resources according to the standard robustness. Specifically, in the resource theory of entanglement there exist states $\rho, \omega$ such that
\begin{equation}\begin{aligned}
  r_{p>0}\big(\rho \toARNGs \omega\big) \,<\, r_{p>0}\big(\omega \toARNGs \rho\big)^{-1}.
\end{aligned}\end{equation}
\end{theorem}

Together with our achievability result in Theorem~\ref{thm:reversibility}, this provides a complete characterisation of the landscape of reversibility in the probabilistic setting: general reversibility is indeed achievable with the asymptotically resource--non-generating transformations $\OO_\ARNG$, and not only is it impossible under operations that are not allowed to create any resources, but even a slightly more restrictive choice of operations obtained by the change of the underlying resource measure from $R^g_\FF$ to $R^s_\FF$ precludes reversibility in general quantum resource theories.

A detailed derivation of Theorem~\ref{thm:irrev_main}, together with technical details and extensions, can be found in Section~\ref{sec:app_irrev} of the~\hyperlink{supp}{Supplementary Information}.


\para{Entanglement distillation} Two of the most important problems in the understanding of asymptotic entanglement manipulation concern the tasks of extracting `entanglement bits' (ebits), i.e.\ copies of the maximally entangled two-qubit singlet state $\Phi_+$, and the reverse task of converting ebits into general noisy states. The rates of these two tasks are known as, respectively, the \emph{distillable entanglement} $E_{d,\OO}^{x} (\rho) \coloneqq r_{x}(\rho \to \Phi_+)$ and the \emph{entanglement cost} $E_{c,\OO}^{x} (\rho) \coloneqq r_{x}(\Phi_+ \to \rho)^{-1}$, where $x$ stands for either $p\!\!=\!\!1$, $p\!>\!0$, or the strong converse rate $p\!\!=\!\!1,\!\dagger$.
Although exact expressions can be obtained for the entanglement cost in various settings~\cite{hayden_2001,brandao_2010,audenaert_2003}, the understanding of distillable entanglement appears to be an extremely difficult problem that has so far resisted most attempts at a conclusive solution, except in some special cases~\cite{devetak_2005, DNE-distillable}. Of note is the conjectured result that~\cite{brandao_2010,berta_2022}
\begin{equation}\begin{aligned}\label{eq:bp_dist_conj}
  E_{d,\rm NE}^{p=1}(\rho) \stackrel{?}{=} D^\infty_{\rm SEP}(\rho),
\end{aligned}\end{equation}
where NE stands for the class of non-entangling operations (equivalent to RNG maps in this theory) and $D^\infty_{\rm SEP}$ is the regularised relative entropy of entanglement.
Establishing this result would recover the deterministic reversibility of entanglement theory (that is, Eq.~\eqref{eq:BP_conjecture})~\cite{brandao_2010,berta_2022}.
We note that distillation rates under NE operations are equal to rates under asymptotically non-entangling operations (ANE)~\cite{lami_2023}, which correspond to $\OO_{\rm ARNG}$ in the notation of this work.

We now introduce a close relation that connects entanglement distillation transformations in the probabilistic and strong converse regimes.
Namely, we show that one can always improve on the transformation error of a distillation protocol by sacrificing some success probability, and vice versa. 
What this means in particular is that every rate that can be achieved in the deterministic strong converse regime (
i.e.\ with a possibly large error $\ve < 1$) can also be achieved probabilistically with error going to zero.
Crucially, to construct the new, modified protocol from the original one we only need to employ local operations and classical communication (LOCC), which are the standard class of free operations in entanglement theory, meaning that the result applies to essentially all different types of operations that extend LOCC.

\begin{theorem}\label{thm:entanglement}
Let $\OO$ be any class of operations which is closed under composition with LOCC, i.e.\ such that $\E \in \OO,\, \F \in \rm LOCC \,\Rightarrow\, \F \circ \E \in \OO$. This includes in particular the set $\rm LOCC$ itself. 
Then, for all states $\rho$,
\begin{equation}\begin{aligned}\label{eq:ent_equivalence}
  E_{d,\OO}^{p=1,\dagger}(\rho) = E_{d,\OO}^{p>0}(\rho).
\end{aligned}\end{equation}
For the case of (asymptotically) non-entangling operations, we have that
\begin{equation}\begin{aligned}\label{eq:ent_equivalence_ne}
  &E_{d,{\rm (A)NE}}^{p=1,\dagger}(\rho) = E_{d,{\rm (A)NE}}^{p>0}(\rho) = D^\infty_{\rm SEP}(\rho).
\end{aligned}\end{equation}
\end{theorem}

\begin{proof}[Proof sketch]
Assume that two spatially separated parties, conventionally called Alice and Bob, share $n$ copies of an entangled state $\rho_{AB}$. 
Consider any sequence of protocols which allows them to distill entanglement from such states at a rate $r$ with with error $\ve_n \tends \ve$ and probability $p_n \tends p$. (This includes the deterministic strong converse case where $p_n=1$.) 
After performing the considered distillation protocol, they share a many-copy state $\tau_{A'B'}$ which approximates $\Phi_+^{\otimes \floor{rn}}$. What they can do now is to sacrifice a fixed number $k$ of their 
qubit pairs in order to perform a state discrimination protocol: by testing whether the $k$-copy subsystem is in the state $\Phi_+^{\otimes k}$ and discarding their whole state when it is not, they can probabilistically bring the state of the rest of their shared system closer to $\Phi_+^{\otimes \floor{rn} - k}$. We show that this can be done by a simple LOCC protocol wherein Alice and Bob perform measurements in the computational basis and compare their outcomes. Since $k$ is arbitrary here, Alice and Bob can perform the modified protocol without reducing the asymptotic transformation rate. 
Conversely, in a similar manner they can also increase their probability of success by sacrificing some transformation fidelity. 
By deriving the exact conditions for when distillation protocols can be refashioned in such a away, we observe that another sequence of entanglement distillation protocols with error $\ve'_n \tends \ve'$ and probability $p'_n \tends p'$ can exist \emph{if and only if}
\begin{equation}\begin{aligned}
  p\, (1-\ve) = p' \, (1-\ve').
\end{aligned}\end{equation}
This directly implies Eq.~\eqref{eq:ent_equivalence}.

To see Eq.~\eqref{eq:ent_equivalence_ne}, it suffices to combine Theorem~\ref{thm:reversibility} with the known result that 
 $D^\infty_{\rm SEP}(\rho)$ is a strong converse rate for distillation under ANE~\cite{hayashi_2006-2,brandao_2010}.
\end{proof}

We have already shown the probabilistic reversibility of entanglement theory in Theorem~\ref{thm:reversibility}, so let us now discuss the deterministic case. Here, reversibility is fully equivalent to the question of whether $E^{p=1}_{d,\OO}(\rho) = E^{p=1}_{c,\OO}(\rho)$ holds for all quantum states, which has been conjectured to be true for the class of asymptotically non-entangling operations~\cite{brandao_2010}. Combining our results with the known findings of Brand\~ao and Plenio~\cite{brandao_2010}, 
we have that
\begin{equation}\begin{aligned}
E_{d,\rm ANE}^{p>0}(\rho) &\overset{\text{Thm.~\ref{thm:entanglement}}}{=\joinrel=} E_{d,\rm ANE}^{p=1,\dagger}(\rho) \overset{\text{\scriptsize\cite{brandao_2010,berta_2022}}}{=\joinrel=} E_{c,\rm ANE}^{p=1}(\rho)\\
&\;\;\overset{\text{\scriptsize\cite{brandao_2010}}}{=\joinrel=} D^\infty_{\rm SEP}(\rho) \overset{\text{Thm.~\ref{thm:reversibility}}}{=\joinrel=} E_{c,\rm ANE}^{p>0}(\rho) .
\end{aligned}\end{equation}
The missing link is thus the question if $E_{d,\rm ANE}^{p=1}(\rho) \stackrel{?}{=} E_{d,\rm ANE}^{p>0}(\rho)$, or the equivalent~\cite{berta_2022} question of whether $E_{c,\rm ANE}^{p=1,\dagger}(\rho) \stackrel{?}{=} E_{c,\rm ANE}^{p>0}(\rho)$. Showing either of these statements would complete the proof of the deterministic reversibility of quantum entanglement under asymptotically non-entangling operations. An interesting consequence of the above is that establishing the equivalent of Theorem~\ref{thm:entanglement} for entanglement \emph{dilution} would be sufficient to recover a fully reversible entanglement theory.

We remark that other quantum resource theories may not be amenable to a characterisation in terms of distillation and dilution because they may not possess a suitably well-behaved unit of a resource resembling the maximally entangled state~\cite{brandao_2015,regula_2020,takagi_2022-1}. Nonetheless, reversibility in all resource theories can be understood as in our Theorem~\ref{thm:reversibility}.


\section*{Discussion} 

We have shown that the conjectured reversibility of general quantum resources can be recovered, albeit in a \emph{probabilistic} manner that employs probabilistic protocols with non-vanishing probability of success.
This allowed us to identify the setting of probabilistic resource transformations as one that is completely governed by a unique entropic quantity --- the regularised relative entropy --- thus solidifying the parallels between thermodynamics and diverse types of quantum resources. We further showed that the choice of asymptotically resource--non-generating operations is optimal in this setting, in the sense that
all smaller classes of probabilistic operations are necessarily irreversible, providing a strong no-go restriction on how reversibility could be achieved.

Although this precise characterisation of asymptotic rates is appealing, our setting departs from the original reversibility conjectures of~\cite{brandao_2010,brandao_2015}, since it considers transformations that are only required to be achieved with some probability. This may not be enough to ensure the existence of a repeatable, reversible transformation cycle in practice. Nevertheless, in view of the close relations between probabilistic and deterministic rates (Eq.~\eqref{eq:rate_hierarchy}) we regard our results as evidence that reversibility could indeed be recovered also in the deterministic setting. This is further motivated by the fact that, in many quantum information processing tasks, strong converse rates actually coincide with the optimal achievable rates~\cite{winter_1999,ogawa_2000,ogawa_1999,konig_2009-1,bennett_2014,berta_2011,morgan_2014,tomamichel_2017}, meaning that $r_{p=1} = r^\dagger_{p=1}$ and the hierarchy in Eq.~\eqref{eq:rate_hierarchy} collapses. However, a complete proof of this fact in the setting of resource manipulation remains elusive, and it is still possible that one of the inequalities in Eq.~\eqref{eq:rate_hierarchy} may be strict for some states, which would rule out deterministic reversibility. 
We hope that our results stimulate further research in this direction, leading to an eventual resolution of the open questions that cloud the understanding of asymptotic resource manipulation.

\subsection*{Acknowledgments}
We are grateful to the anonymous referees of the conference Quantum Information Processing 2024 for helpful comments. 
We acknowledge discussions with Tulja Varun Kondra and Alexander Streltsov. We further thank the Freie Universit\"at Berlin for hospitality.


\let\oldaddcontentsline\addcontentsline
\renewcommand{\addcontentsline}[3]{}
 \bibliographystyle{apsc}
 \bibliography{main}
\let\addcontentsline\oldaddcontentsline


\onecolumngrid
\clearpage
\newgeometry{left=1.2in,right=1.2in,top=.7in,bottom=1in}

\stepcounter{part}
\counterwithin{theorem}{part}

\renewcommand{\theequation}{S\arabic{equation}}
\renewcommand{\thetheorem}{S\arabic{theorem}}
\setcounter{equation}{0}
\setcounter{theorem}{0}

\hypertarget{supp}{}
\begin{center}
\vspace*{\baselineskip}
{\textbf{\large --- Supplementary information ---}}\\[1pt] \quad \\
\end{center}

\setcounter{tocdepth}{0}
\tableofcontents

\section{Notation and basic properties}

We begin by clarifying the assumptions and definitions used in this work.

\subsection{Free operations and free states}\label{sec:axioms}

Throughout this paper we assume that the underlying Hilbert space has finite dimension. 
When discussing many-copy state transformations, we assume that each $n$-copy space of quantum states has its own associated set of \deff{free states} $\FF_n$. We often use $\FF$ to refer to the whole family $(\FF_n)_n$ for simplicity.

Following~\cite{brandao_2010-1,brandao_2015}, we will assume the following basic axioms about the resource theory in consideration:
\begin{enumerate}[{Axiom }I.]
\item Each $\FF_n$ is convex and closed.
\item There exists a full-rank state $\sigma$ such that $\sigma^{\otimes n} \in \FF_n$ for all $n$.
\item The sets $\FF_n$ are closed under partial trace: if $\sigma \in \FF_{n+1}$, then $\Tr_{k} \sigma \in \FF_{n}$ for every $k \in \{1,\ldots,n+1\}$.
\item The sets $\FF_n$ are closed under tensor product: if $\sigma \in \FF_n$ and $\sigma' \in \FF_{m}$, then $\sigma \otimes \sigma' \in \FF_{n+m}$.
\end{enumerate}
The motivation for these assumptions is so that the theory is well behaved asymptotically~\cite{brandao_2010-1,datta_2009-2,brandao_2015}.\footnote{In~\cite{brandao_2010-1,brandao_2015}, an additional axiom of permutation invariance was imposed. We do not need this here.} 
These undemanding conditions are obeyed by the vast majority of practically relevant resource theories. In addition, we will need one more assumption that was implicitly used in~\cite{brandao_2010-1} and later appeared explicitly in~\cite{brandao_2015}:
\begin{enumerate}[{Axiom }I.]
\setcounter{enumi}{4}
\item The regularised relative entropy $D^\infty_\FF(\rho) = \lim_{n\to\infty}\frac1n \min_{\sigma\in\FF_n}D(\rho^{\otimes n}\|\sigma)$ is non-zero for all states $\rho \notin \FF$.
\end{enumerate}
Once again, Axiom~V holds true for most resources encountered in practice, in particular for quantum entanglement~\cite{piani_2009-1}. However, there do exist theories in which $D^\infty_\FF$ vanishes for some (or even all) states $\rho \notin \FF$, e.g.\ the theory of asymmetry~\cite{gour_2009}. The standard reversibility conjectures do not apply to such theories~\cite{brandao_2015}.

We use CPTP to denote completely positive and trace-preserving maps (quantum channels), and CPTNI to denote completely positive and trace--non-increasing maps (probabilistic quantum operations). A \deff{probabilistic protocol} is a collection of CPTNI maps $\{\E^{(i)}\}_i$ that forms a valid quantum instrument, i.e.\ the overall transformation $\sum_i \E^{(i)}$ is trace preserving. The latter condition simply means that the outcome probabilities necessarily add up to one. Given some class of \deff{free operations} $\OO$, we say that a protocol is free if $\E^{(i)} \in \OO \; \forall i$. When discussing transformations between two fixed states $\rho \to \omega$, we may coarse-grain the outcomes of the instrument and consider only two-outcome protocols $\{\E^{(i)}\}_{i \in \{0,1\}}$ --- either the protocol is successful in performing the given transformation, or it fails. We can then say that $\rho$ can be converted to $\omega$ probabilistically if there exists a free protocol $\{\E^{(i)}\}_{i \in \{0,1\}}$ such that $\E^{(0)} (\rho) \propto \omega$.

\deff{Resource non-generating (RNG) operations} are the maps which can never transform a resourceless state into a resourceful one:
\begin{equation}\begin{aligned}
  \OO_\RNG \coloneqq \lsetr \E \in \CPTNI \barr \frac{\E(\sigma)}{\Tr \E(\sigma)} \in \FF \; \forall \sigma \in \FF \rsetr.
\end{aligned}\end{equation}
Their approximate variant is defined as
\begin{equation}\begin{aligned}
  \OO_{\RNG,\delta} \coloneqq \lsetr \E \in \CPTNI \barr R^g_\FF\left(\frac{\E(\sigma)}{\Tr \E(\sigma)}\right) \leq \delta\; \forall \sigma \in \FF \rsetr.
\end{aligned}\end{equation}
We note here that the latter definition depends on the choice of the measure with which one quantifies the generated resources --- in this case, this is the \deff{generalised robustness}~\cite{vidal_1999}
\begin{equation}\begin{aligned}
  R^g_\FF(\rho) &\coloneqq  \inf \lsetr \lambda \in \RR_+ \barr \frac{\rho + \lambda \omega}{1+\lambda} \in \FF,\; \omega \in \DD \rsetr\\
  &\hphantom{:}= \inf \lsetr \lambda \in \RR_+ \barr \rho \leq (\lambda+1) \sigma,\; \sigma \in \FF \rsetr.
\end{aligned}\end{equation}
This quantity can also be identified with the max-relative entropy of a resource~\cite{datta_2009-2}, $D_{\max,\FF}(\rho) \coloneqq \log \left( 1 + R^g_\FF(\rho)\right)$. 
Other choices of measures may lead to completely different asymptotic behaviour~\cite{lami_2023}, and our choice here is motivated by the framework of~\cite{brandao_2008-1,brandao_2010} that conjectured reversibility of quantum resources under operations defined in this way.

Under resource non-generating transformations $\OO_\RNG$ (or $\OO_{\RNG,\delta}$), the existence of a probabilistic protocol that performs a transformation $\rho \to \omega$ is fully equivalent to the existence of a map $\E \in \OO_\RNG$ such that $\E(\rho)/\Tr \E(\rho) = \omega$. This is because any such map can always be completed to a free instrument: we simply define $\E'(X) \coloneqq \left[\Tr X - \Tr \E(X)\right] \sigma$ for some $\sigma \in \FF$, yielding $(\E, \E') \in \OO$. For other sets of operations, it may be the case that a map $\E \in \OO$ cannot always be completed to a free instrument $(\E, \E')$ --- e.g.\ for separable operations in entanglement theory~\cite{theurer_2017} or stabiliser operations in the theory of magic states~\cite{seddon_2022} --- so more care may need to be taken when discussing probabilistic transformations.

\subsection{Transformation rates}\label{sec:rates}

Given a class of completely positive and trace non-increasing maps $\OO$, the various transformation rates that we study are defined as follows. 

The \deff{deterministic conversion rate} is
\begin{equation}\begin{aligned} \label{eq:def_rate}
        r_{p=1}(\rho \toO \omega) \coloneqq \sup_{(\E_n)_n}  \Bigg\{ \,r \;\Bigg|\; & \lim_{n\to\infty} F\!\left( \E_n(\rho^{\otimes n}), \, \omega^{\otimes \floor{rn}} \right) = 1,\quad \E_n \in \OO \cap {\rm CPTP} \Bigg\},
\end{aligned}\end{equation}
where the optimisation is over all sequences $(\E_n)_n$ of maps satisfying the given constraints.
Here, $F(\rho,\sigma) \coloneqq \norm{\sqrt{\rho}\sqrt{\vphantom{\rho}\sigma}}{1}^2$ is the fidelity. 

The \deff{probabilistic rate with non-vanishing probability of success} is
\begin{equation}\begin{aligned} \label{eq:def_prob_rate}
    r_{p>0}(\rho \toO \omega) \coloneqq \sup_{(\E_n)_n}  \Bigg\{ \,r \;\Bigg|\; & \lim_{n\to\infty} F\!\left( \frac{\E_n(\rho^{\otimes n})}{\Tr \E_n(\rho^{\otimes n})} ,\, \omega^{\otimes \floor{rn}} \right) = 1,\\
     & \E_n \in \OO,\quad \liminf_{n\to\infty}\, \Tr \E_n(\rho^{\otimes n}) > 0 \Bigg\}.
\end{aligned}\end{equation}
We refer to $1 - F\!\left( \frac{\E_n(\rho^{\otimes n})}{\Tr \E_n(\rho^{\otimes n})} ,\, \omega^{\otimes \floor{rn}} \right)$ as the \deff{transformation error} and to $\Tr \E_n(\rho^{\otimes n})$ as the \deff{transformation probability}. 

The \deff{deterministic strong converse rate} is
\begin{equation}\begin{aligned}\label{eq:def_sc}
    r_{p=1}^\dagger(\rho \toO \omega) &\coloneqq \sup_{(\E_n)_n} 
    \Bigg\{ \,r \;\Bigg|\; \liminf_{n\to\infty}\,  F\!\left( \E_n(\rho^{\otimes n}) ,\, \omega^{\otimes \floor{rn}} \right) > 0,\quad \E_n \in \OO \cap {\rm CPTP} \Bigg\}.\\
    &\hphantom{:}= \inf \Bigg\{ \,r' \;\Bigg|\;  \liminf_{n\to\infty} \sup_{\E_n\in \OO\, \cap\, {\rm CPTP}} F\!\left( \E_n(\rho^{\otimes n}),\, \omega^{\otimes \floor{r'n}} \right) = 0
    \Bigg\}.
\end{aligned}\end{equation}
We note here that there are differences in the precise definition of a strong converse in the literature, and stronger definitions can sometimes be encountered. One such variant is
\begin{equation}\begin{aligned}\label{eq:def_sc_stronger}
    r_{p=1}^\ddagger(\rho \toO \omega) &\coloneqq \sup_{(\E_n)_n}
    \Bigg\{ \,r \;\Bigg|\; \limsup_{n\to\infty}\,  F\!\left( \E_n(\rho^{\otimes n}) ,\, \omega^{\otimes \floor{rn}} \right) > 0, \quad \E_n \in \OO \cap {\rm CPTP} \Bigg\}\\
    &\hphantom{:}= \inf \Bigg\{ \,r' \;\Bigg|\;  \lim_{n\to\infty} \sup_{\E_n\in \OO \,\cap\, {\rm CPTP}} F\!\left( \E_n(\rho^{\otimes n}),\, \omega^{\otimes \floor{r'n}} \right) = 0 
    \Bigg\}.
\end{aligned}\end{equation}
The difference between $r^\dagger$ and $r^\ddagger$ is that the former maximises over rates which are truly achievable, albeit with a large error; in other words, for all sufficiently large $n$, the fidelity with the target state must be larger than some non-zero constant. On the other hand, the maximisation in~\eqref{eq:def_sc_stronger} only requires that the fidelity be non-zero infinitely often.
It can be noticed that $r_{p=1}^\dagger \leq r_{p=1}^\ddagger$ in general, although we do not know if this inequality can ever be strict when considering i.i.d.\ state transformations as we do here.
We choose to employ the definition of $r^\dagger$ as in Eq.~\eqref{eq:def_sc} as it will lead to tighter statements of our results. On the other hand, we will sometimes use $r^\ddagger$ to provide potentially stronger upper (converse) bounds on rates. A reader not interested in these technicalities may safely ignore the difference between the two definitions.

For completeness, we will also briefly consider the \deff{strong converse probabilistic rate}
\begin{equation}\begin{aligned} \label{eq:def_prob_sc}
r^\dagger_{p>0}(\rho \toO \omega) \coloneqq \sup_{(\E_n)_n}
\Bigg\{ \,r \;\Bigg|\; & \liminf_{n\to\infty} \, F\!\left( \frac{\E_n(\rho^{\otimes n})}{\Tr \E_n(\rho^{\otimes n})} ,\, \omega^{\otimes \floor{rn}} \right) > 0,\\
& \E_n \in \OO,\quad \liminf_{n\to\infty}\, \Tr \E_n(\rho^{\otimes n}) > 0 \Bigg\} ,
\end{aligned}\end{equation}
together with its variant
\bb \label{eq:def_prob_sc_stronger}
r_{p>0}^\ddagger(\rho \toO \omega) \coloneqq \sup_{(\E_n)_n}
    \Bigg\{ \,r \;\Bigg|\; &\limsup_{n\to\infty}\,  F\!\left( \frac{\E_n(\rho^{\otimes n})}{\Tr \E_n(\rho^{\otimes n})} ,\, \omega^{\otimes \floor{rn}} \right) > 0, \\ 
    &\E_n \in \OO,\quad \liminf_{n\to\infty}\, \Tr \E_n(\rho^{\otimes n}) > 0 \Bigg\} ,
\ee
where the only difference between Eq.~\eqref{eq:def_prob_sc} and Eq.~\eqref{eq:def_prob_sc_stronger}, just like for Eq.~\eqref{eq:def_sc} and Eq.~\eqref{eq:def_sc_stronger}, is the presence of a $\liminf$ or a $\limsup$ over the achievable transformation fidelities.

\setlength{\arrayrulewidth}{0.2mm}
\setlength{\tabcolsep}{9pt}
\renewcommand{\arraystretch}{1.8}

\begin{table}
    \centering
    \begin{tabular}{c|c|c|c} 
    & Standard & Strong converse $\dag$ & Stronger strong converse $\ddagger$ \\ \hline
    Deterministic & $r_{p=1}(\rho \toO \omega)$~~\eqref{eq:def_rate} & $r_{p=1}^\dagger(\rho \toO \omega)$~~\eqref{eq:def_sc} & $r_{p=1}^\ddagger(\rho \toO \omega)$~~\eqref{eq:def_sc_stronger} \\
    Probabilistic & $r_{p>0}(\rho \toO \omega)$~~\eqref{eq:def_prob_rate} & $r^\dagger_{p>0}(\rho \toO \omega)$~~\eqref{eq:def_prob_sc} & $r^\ddagger_{p>0}(\rho \toO \omega)$~~\eqref{eq:def_prob_sc_stronger}
    \end{tabular}
    \caption{The transformation rates defined by Eqs.~\eqref{eq:def_rate}--\eqref{eq:def_prob_sc_stronger}.}
\end{table}

Finally, our main object of study are probabilistic transformation rates under \deff{asymptotically resource--non-generating transformations} $\OO_{\rm ARNG}$, which are sequences of operations in $\OO_{\RNG,\delta_n}$ such that $\delta_n \tends 0$. Specifically,
\begin{equation}\begin{aligned}\label{eq:def_rate_arng}
    r_{p>0}(\rho \toARNG \omega) \coloneqq \sup_{(\E_n)_n}  \Bigg\{ \,r \;\Bigg|\; & \lim_{n\to\infty} F\!\left( \frac{\E_n(\rho^{\otimes n})}{\Tr \E_n(\rho^{\otimes n})} ,\, \omega^{\otimes \floor{rn}} \right) = 1,\\
    & \lim_{n\to\infty} \,\sup_{\sigma \in \FF_n} \,R^g_{\FF_{\floor{rn}}} \!\left(\frac{\E_n(\sigma)}{\Tr \E_n(\sigma)}\right)  = 0,\quad \liminf_{n\to\infty}\, \Tr \E_n(\rho^{\otimes n}) > 0 \Bigg\}.
\end{aligned}\end{equation}

When employing the above definitions, we will make use of the fact that, due to the Fuchs--van de Graaf inequalities~\cite{fuchs_1999}, the transformation error can be equivalently defined with the trace distance~$\frac12\norm{\rho-\sigma}{1}$:
\begin{equation}\begin{aligned}
  \lim_{n\to\infty} \, F\!\left( \frac{\E_n(\rho^{\otimes n})}{\Tr \E_n(\rho^{\otimes n})}, \, \omega^{\otimes \floor{rn}} \right) = 1 \; &\iff \; \lim_{n\to\infty}\, \frac12 \norm{ \frac{\E_n(\rho^{\otimes n})}{\Tr \E_n(\rho^{\otimes n})} - \omega^{\otimes \floor{rn}}}{1} = 0,\\
  \liminf_{n\to\infty} \,F\!\left( \frac{\E_n(\rho^{\otimes n})}{\Tr \E_n(\rho^{\otimes n})}, \, \omega^{\otimes \floor{rn}} \right) > 0 \; &\iff \; \limsup_{n\to\infty} \frac12 \norm{ \frac{\E_n(\rho^{\otimes n})}{\Tr \E_n(\rho^{\otimes n})} - \omega^{\otimes \floor{rn}}}{1} < 1.
\end{aligned}\end{equation}

The relations between the different transformation rates are encapsulated in the following lemma.
\begin{boxed}
\begin{lemma}\label{lem:hierarchy}
Let $\OO$ be any class of free operations such that, for an instrument $\{\E^{(i)}\}_i$ with $\E^{(i)} \in \OO$, it holds that $\sum_i \E^{(i)} \in \OO$. This can be, for instance, the classes $\OO_{\rm RNG}$ or $\OO_{\rm RNG, \delta}$ in any convex resource theory, or any suitable class of free operations in entanglement theory such as LOCC.

Then, for all quantum states $\rho$ and $\omega$, the transformation rates satisfy
    \begin{equation}\begin{aligned}
        r_{p=1}(\rho \toO \omega) \,\leq\, r_{p > 0}(\rho \toO \omega) \,&\leq\, r^\dagger_{p=1}(\rho \toO \omega) \,=\, r^\dagger_{p>0}(\rho \toO \omega) \\
        &\leq\, r^\ddagger_{p=1}(\rho \toO \omega) \,=\, r^\ddagger_{p>0}(\rho \toO \omega).
    \end{aligned}\end{equation}
\end{lemma}
\end{boxed}

\begin{remark}Let us comment briefly on the assumption of this lemma. If there exists a classical system $C$ with an orthonormal basis $\{\ket{i}_C\}_i$ such that each element of the basis is a free state, 
one can equivalently understand the realisation of any probabilistic protocol  $\{\E^{(i)}\}_i$ as a CPTP map
\begin{equation}\begin{aligned}
  \rho \mapsto \sum_i \E^{(i)}(\rho) \otimes \proj{i}_C
\end{aligned}\end{equation}
where the classical system is used to keep track of the probabilistic outcome of the protocol. If we are free to `forget' this outcome --- that is, if tracing out the classical system is an allowed operation --- then the given set $\OO$ satisfies the conditions of the lemma. This is a very weak assumption that is satisfied by essentially all types of free operations in the majority of resource theories, but there are special cases in which it may not be obeyed, e.g.\ for thermal operations in the theory of thermodynamics~\cite{alhambra_2016}.
\end{remark}

\begin{proof}
    As deterministic protocols are a special case of probabilistic ones, the first inequality is immediate. For the second inequality, assume that there exists a sequence of free probabilistic operations $(\E_n)_n$ which performs the transformation $\rho \to \omega$ at a rate $r$ with error $\ve_n$ such that $\lim_{n\to\infty}\ve_n = 0$ and probability $p_n$ such that $\liminf_{n\to\infty} p_n = p > 0$. 
    By definition of a probabilistic protocol, there exist operations $\E'_n \in \OO$ which complete each $\E_n$ to a valid quantum instrument, that is, such that
\begin{equation}\begin{aligned}
    \D_n(X) \coloneqq \E_n(X) + \E'_n(X)
\end{aligned}\end{equation}
is trace preserving. (When $\OO$ denotes resource non-generating operations, we can simply take $\E'_n(X) = \left[\Tr X - \Tr \E_n(X)\right] \sigma_n$ for some $\sigma_n \in \FF$.)\; By hypothesis, we have that $\D_n \in \OO$. 

Since $p_n = \Tr \E_n(\rho^{\otimes n})$, the concavity of the fidelity then gives
\begin{equation}\begin{aligned}\label{eq:fidelity_prob_combined}
    \liminf_{n\to\infty} \, F\!\left(\D_n(\rho^{\otimes n}), \omega^{\otimes \floor{rn}}\right) &\geq \liminf_{n\to\infty} \, p_n \, F\! \left( \frac{\E_n(\rho^{\otimes n})}{p_n},\, \omega^{\otimes \floor{rn}} \right)\\
    &= \liminf_{n\to\infty} \, p_n (1-\ve_n)\\
    &= p > 0,
\end{aligned}\end{equation}
hence $r \leq r^\dagger_{p=1}(\rho \to \omega)$.

Let now $r$ be any feasible rate for $r^\dagger_{p>0}(\rho \toO \omega)$, that is, such that there exists a sequence of free probabilistic operations $(\E_n)_n$ with non-vanishing error $\limsup_{n\to\infty} \ve_n = \ve < 1$ and probability of success $\liminf_{n\to\infty} p_n = p > 0$. The same argument as above can be used to construct a deterministic protocol with
\begin{equation}\begin{aligned}
  \liminf_{n\to\infty}\, F\!\left(\D_n(\rho^{\otimes n}), \omega^{\otimes \floor{rn}}\right) &\geq \liminf_{n\to\infty} \, p_n (1-\ve_n) \geq p (1-\ve) > 0 ,
\end{aligned}\end{equation}
yielding $r \leq r^\dagger_{p=1}(\rho \toO \omega)$ and hence $r^\dagger_{p>0} \leq r^\dagger_{p=1}$. Since the opposite inequality between the strong converse rates follows straightforwardly from the fact that any deterministic protocol is a special case of a probabilistic one, the two rates must be equal.

The inequality $r^\dagger_{p>0} \leq r^\ddagger_{p>0}$ follows immediately by comparing~\eqref{eq:def_prob_sc} and~\eqref{eq:def_prob_sc_stronger}. The equivalence between the stronger strong converse rates $r^\ddagger_{p>0}$ and $r^\ddagger_{p=1}$ proceeds analogously, where we now use that
\begin{equation}\begin{aligned}
  \limsup_{n\to\infty}\, F\!\left(\D_n(\rho^{\otimes n}), \omega^{\otimes \floor{rn}}\right) &\geq \liminf_{n\to\infty} \, p_n \, \limsup_{n\to\infty} (1-\ve_n) \geq p (1-\ve) > 0 .
\end{aligned}\end{equation}
This completes the proof.
\end{proof}


\section{Probabilistic reversibility of general quantum resources}\label{sec:app_rev}

{
\renewcommand{\thetheorem}{1}
\begin{boxed}[filled]
\begin{theorem} 
\label{thm:reversibility-s}
In any resource theory satisfying axioms I--V of Sec.~\ref{sec:axioms}, it holds that
\begin{equation}\begin{aligned}
    r_{p>0}(\rho \toARNG \omega) = \frac{D^\infty_\FF(\rho)}{D^\infty_\FF(\omega)}.
\end{aligned}\end{equation}
\end{theorem}
\end{boxed}
}

\begin{proof}
\parait{Converse} 
We will often use the logarithmic variant of $R^g_\FF$, for which we recall the standard notation
\begin{equation}\begin{aligned}
D_{\max,\FF}(\rho) = \log \left( 1 + R^g_\FF(\rho)\right).
\end{aligned}\end{equation}
Assume that $r$ is an achievable probabilistic rate, that is, there exists a sequence of protocols $(\E_n)_n$ with $\E_n \in \OO_{\RNGdeltan}$ such that $\frac12 \norm{\frac{\E_n(\rho^{\otimes n})}{\Tr \E_n(\rho^{\otimes n})} - \omega^{\otimes \floor{rn}}}{1} \eqqcolon \ve_n$ with $\lim_{n\to\infty} \ve_n = 0$, $\liminf_{n\to \infty} \Tr \E_n(\rho^{\otimes n}) \eqqcolon p > 0$, and furthermore $\lim_{n\to\infty} \delta_n = 0$.
For all $\zeta<p$ and all sufficiently large $n$, it then holds that $\zeta/\Tr \E_n(\rho^{\otimes n}) + \ve_n < 1$. 
Letting
\begin{equation}\begin{aligned}
   D^\zeta_{\max,\FF} (\rho) \coloneqq \min_{\frac12 \norm{\rho - \rho'}{1} \leq \zeta} D_{\max,\FF}(\rho'),
\end{aligned}\end{equation}
we use the probabilistic monotonicity of the max-relative entropy (Lemma~\ref{lem:dmax} below) to get 
\begin{equation}\begin{aligned}
   D^\zeta_{\max,\FF}(\rho^{\otimes n}) &\geq D^{\zeta/\Tr \E_n(\rho^{\otimes n})}_{\max,\FF}\left(\frac{\E_n(\rho^{\otimes n})}{\Tr \E_n(\rho^{\otimes n})}\right) + \log \left( \Tr \E_n(\rho^{\otimes n}) - \zeta\right) - \log (1+\delta_n)\\
&\geq D^{\zeta/\Tr \E_n(\rho^{\otimes n})+\ve_n}_{\max,\FF}\left(\omega^{\otimes \floor{rn}}\right) + \log \left( \Tr \E_n(\rho^{\otimes n}) - \zeta\right) - \log (1+\delta_n).
\end{aligned}\end{equation}
Here, in the last line we used the fact that $\frac12\norm{\tau - \omega^{\otimes \floor{rn}}}{1} \leq \frac12\norm{\tau - \frac{\E_n(\rho^{\otimes n})}{\Tr \E_n(\rho^{\otimes n})}}{1} + \frac12 \norm{\frac{\E_n(\rho^{\otimes n})}{\Tr \E_n(\rho^{\otimes n})} - \omega^{\otimes \floor{rn}}}{1}$ for any state $\tau$.

Applying the above, we obtain that
\begin{equation}\begin{aligned}
    \inf_{(\zeta_n)_n} \lset \limsup_{n\to\infty} \frac1n D^{\zeta_n}_{\max,\FF}(\rho^{\otimes n}) \bar \lim_{n\to\infty} \zeta_n = 0 \rset 
    &\geq \inf_{(\zeta_n)_n} \lset \limsup_{n\to\infty} \frac1n D^{\zeta_n/\Tr \E_n(\rho^{\otimes n})+\ve_n}_{\max,\FF}\left(\omega^{\otimes \floor{rn}}\right) \bar \lim_{n\to\infty} \zeta_n = 0 \rset\\
&\geq \inf_{(\eta_n)_n} \lset \limsup_{n\to\infty} \frac1n D^{\eta_n}_{\max,\FF}\left(\omega^{\otimes \floor{rn}}\right) \bar \lim_{n\to\infty} \eta_n = 0\rset,
\label{eq:converse_mother_inequality}
\end{aligned}\end{equation}
where in the third line we observed that $\zeta_n/\Tr \E_n(\rho^{\otimes n}) + \ve_n$ tends to $0$ for any sequence $(\zeta_n)_n$ tending to $0$. 
Using the asymptotic equipartition property~\cite{brandao_2010-1,datta_2009-2}
\begin{equation}\begin{aligned}
    \inf_{(\zeta_n)_n} \lset \limsup_{n\to\infty} \frac1n D^{\zeta_n}_{\max,\FF}(\omega^{\otimes \floor{rn}}) \bar \lim_{n\to\infty} \zeta_n = 0 \rset = r\, D^\infty_\FF(\omega)
\end{aligned}\end{equation}
on both sides of~\eqref{eq:converse_mother_inequality} concludes the converse part of the Theorem. 

\parait{Achievability} From the results of Brand\~ao--Plenio~\cite{brandao_2010-1}, and in particular from the aforementioned equipartition property, we know two things:
\begin{enumerate}[1)]
\item (\cite[Proposition~II.1 and Corollary~III.2]{brandao_2010-1}; also~\cite[Theorem~1]{datta_2009-2} ) For any $\omega$, any $r > 0$ and any $c > D^\infty_\FF(\omega)$, there exists a sequence of states $(\omega_n)_n$ such that $F\!\left(\omega_n,\, \omega^{\otimes \floor{r n}}\right) \geq 1-\ve_n$ with $\lim_{n\to\infty} \ve_n = 0$ and such that
\begin{equation}\begin{aligned}\label{eq:stein_dmax_aep}
    \limsup_{n\to\infty} \frac1n \log \left( 1 + R^g_\FF(\omega_n) \right) = \limsup_{n\to\infty} \frac1n D_{\max,\FF}(\omega_n) = cr.
\end{aligned}\end{equation}

\item (\cite[Corollary~III.3; cf.\ also proof of Theorem I therein]{brandao_2010-1}) For any $\rho$ and any $d < D^\infty_\FF(\rho)$, there exists a sequence of POVM elements $(A_n)_n$ satisfying $0\leq A_n\leq \id$ for all $n$, and moreover $\Tr \left( A_n \rho^{\otimes n} \right) \geq 1-\delta_n$, where the sequence $(\delta_n)_n$ is such that
\begin{equation}\begin{aligned}
\limsup_{n\to\infty} \delta_n = \delta < 1
\end{aligned}\end{equation}
and\footnote{The proof of~\cite{brandao_2010-1} actually makes the stronger statement that we can take $- \frac1n \log \sup_{\sigma \in \FF_n} \Tr( A_n  \sigma) = d$ for all $n$. We do not need this here, hence we leave the statement in the above more general form.}
\begin{equation}\begin{aligned}\label{eq:stein_achiev}
    \liminf_{n\to\infty} - \frac1n \log \sup_{\sigma \in \FF_n} \Tr( A_n  \sigma) = d.
\end{aligned}\end{equation}
\end{enumerate}

We define the shorthands
\begin{equation}\begin{aligned}
    a_n &\coloneqq \sup_{\sigma \in \FF_n} \Tr (A_n \sigma),\\
    \lambda_n &\coloneqq 1 + R^g_\FF(\omega_n).
    \label{eq:shorthands}
\end{aligned}\end{equation}

Let us use the above to construct a probabilistic protocol which transforms $n$ copies of $\rho$ to $\floor{rn}$ copies of $\omega$ for some $r$. We start by defining the sequence of maps
\begin{equation}\begin{aligned}
    \E_n(X) \coloneqq \Tr (A_n X)\, \omega_n + \mu_n  \Tr [ (\id - A_n) X ]\, \pi_n,
\end{aligned}\end{equation}
where $\mu_n \in [0,1]$ are some parameters to be chosen later, $\omega_n$ are the states appearing in~\eqref{eq:stein_dmax_aep}, and $\pi_n$ are states such that
\begin{equation}\begin{aligned}\label{eq:pi_state}
   \frac{ \omega_n + (\lambda_n - 1) \pi_n }{ \lambda_n } = \sigma_n \in \FF_{\floor{rn}},
\end{aligned}\end{equation}
which exist by definition of $R^g_\FF$. The maps are clearly completely positive and trace non-increasing by construction.

We will now show that, for any choice of a sequence $(\mu_n)_n$ such that
\begin{equation}\begin{aligned}\label{eq:mu_condition}
    \mu_n \geq \frac{\lambda_n - 1}{a_n^{-1} - 1}
\end{aligned}\end{equation}
for all $n$, the sequence of maps 
$(\E_n)_n$ is asymptotically resource non-generating.
This can be seen by observing that, for any $\sigma \in \FF_n$, it holds that
\begin{equation}\begin{aligned}
    \E_n(\sigma) &= \Tr (A_n \sigma)\, \omega_n + \mu_n (1 - \Tr (A_n \sigma)\,) \pi_n\\
    &= \Tr (A_n \sigma)\, \lambda_n \frac{ \omega_n + (\lambda_n - 1) \pi_n }{ \lambda_n } + \big[ \mu_n \big( 1 - \Tr (A_n \sigma) \big) - \Tr (A_n \sigma)\, ( \lambda_n - 1 ) \big] \pi_n\\
    &\textleq{(i)} \Tr (A_n \sigma)\, \lambda_n \, \sigma_n + \big[ \mu_n \big( 1 - \Tr (A_n \sigma) \big) - \Tr (A_n \sigma)\, ( \lambda_n - 1 ) \big] \frac{\lambda_n}{\lambda_n - 1} \sigma_n\\
    &\textleq{(ii)} \Tr (A_n \sigma)\, \lambda_n  
    \frac{\lambda_n}{\lambda_n-1}\, \sigma_n + \big[ \mu_n \big( 1 - \Tr (A_n \sigma)\big) - \Tr (A_n \sigma)\, ( \lambda_n - 1 ) \big] \frac{\lambda_n}{\lambda_n - 1} \sigma_n\\
    &= \big[ \mu_n \big(1 - \Tr(A_n \sigma)\big) + \Tr(A_n \sigma) \big] \frac{\lambda_n}{\lambda_n - 1} \sigma_n\\
    &= \big[\!\Tr \E_n(\sigma)\big]\, \frac{\lambda_n}{\lambda_n - 1} \sigma_n.
\end{aligned}\end{equation}
Here, in (i) we used the fact that $\pi_n \leq \frac{\lambda_n}{\lambda_n - 1} \sigma_n$, which follows from Eq.~\eqref{eq:pi_state}, as well as that $\mu_n \big( 1 - \Tr (A_n \sigma)\big) - \Tr (A_n \sigma)\, ( \lambda_n - 1 ) \geq 0$, which is due to the assumption in Eq.~\eqref{eq:mu_condition}. In (ii), we simply used that $\frac{\lambda_n}{\lambda_n - 1} \geq 1$.
The above guarantees that
\begin{equation}\begin{aligned}
    \lim_{n\to\infty} \, R^g_\FF\left(\frac{\E_n(\sigma)}{\Tr \E_n(\sigma)} \right) \leq 
    \lim_{n\to\infty} \left( \frac{\lambda_n}{\lambda_n-1}-1\right) = 0,
\end{aligned}\end{equation}
where the last identity follows from the observation that $\lambda_n$ diverges exponentially, as is apparent from~\eqref{eq:stein_dmax_aep} and~\eqref{eq:shorthands}.

Crucially, the probability of success of these operations satisfies
\begin{equation}\begin{aligned}
    p_n = \Tr \E_n(\rho^{\otimes n}) &= \Tr \left(A_n \rho^{\otimes n}\right) + \mu_n \left(1 - \Tr \left(A_n \rho^{\otimes n}\right) \right)\\
    &\geq 1 - \delta_n (1-\mu_n),
\end{aligned}\end{equation}
while the transformation fidelity is at least
\begin{equation}\begin{aligned}\label{eq:fidelity_prob_combined_proof}
    F\left(\frac{\E_n(\rho^{\otimes n})}{\Tr \E_n(\rho^{\otimes n})}, \omega^{\otimes \floor{rn}} \right) &\geq \frac{\Tr \left(A_n \rho^{\otimes n}\right)}{\Tr \E_n(\rho^{\otimes n})} \, F(\omega_n, \omega^{\otimes \floor{rn}})\\
    &\geq \frac{1-\delta_n}{1 - \delta_n + \mu_n\delta_n} (1-\ve_n),
\end{aligned}\end{equation}
where the first line follows by the concavity of the fidelity.
It follows that, by decreasing $\mu_n$, we can decrease the error in this transformation; on the other hand, the probability of success will never go below $1 - \delta > 0$ no matter how much $\mu_n$ is decreased, since $\delta_n \tends \delta < 1$.

Let us then fix
\vspace*{-.5\baselineskip}
\begin{equation}\begin{aligned}\label{eq:mu_choice}
    \mu_n \coloneqq \frac{\lambda_n - 1}{a_n^{-1} - 1}
\end{aligned}\vspace*{-.5\baselineskip}\end{equation}
and pick a rate
\vspace*{-.5\baselineskip}
\begin{equation}\begin{aligned}
    r = \frac{D^\infty_\FF(\rho)}{D^\infty_\FF(\omega)} -  \chi 
\end{aligned}\end{equation}
for some $\chi > 0$.
For all $\xi,\xi' > 0$ and for all large enough $n$, it then holds that
\begin{equation}\begin{aligned}
    \log \lambda_n &\leq r n\, D^\infty_\FF(\omega) + n \xi\\
    &\leq n\, D^\infty_\FF(\rho) - n \chi D^\infty_\FF(\rho) + n \xi 
\end{aligned}\end{equation}
and
\begin{equation}\begin{aligned}
    \log a_n^{-1} &\geq n\, D^\infty_\FF(\rho) - n \xi'.
\end{aligned}\end{equation}
Picking $\xi' = \xi = \chi D^\infty_\FF(\rho) / 4$ , we get
\begin{equation}\begin{aligned}
    \lim_{n\to\infty} \, \mu_n &= \lim_{n\to\infty} \, \frac{\lambda_n - 1}{a_n^{-1} - 1} \\
    &\leq \lim_{n\to\infty}\, \frac{2^{n D^\infty_\FF(\rho) \left( 1- \chi + \chi/4 \right)} - 1}{2^{n D^\infty_\FF(\rho) \left(1 - \chi/4 \right)} - 1}\\ 
    &= \lim_{n\to\infty}\, \frac{2^{n D^\infty_\FF(\rho) \left( 1- \chi/4 \right) - n D^\infty_\FF(\rho) \chi / 2 } - 1}{2^{n D^\infty_\FF(\rho) \left(1 - \chi/4 \right)} - 1}\\
    &= \lim_{n\to\infty}\,2^{- n D^\infty_\FF(\rho) \chi / 2 }\\
    &= 0.
\end{aligned}\end{equation}

Altogether, the probabilistic protocol $(\E_n)_n$ with the choice of $\mu_n$ as in Eq.~\eqref{eq:mu_choice} performs the transformation $\rho \to \omega$ at the rate $\frac{D^\infty_\FF(\rho)}{D^\infty_\FF(\omega)} - \chi$ with error satisfying
\begin{equation}\begin{aligned}
    \limsup_{n\to\infty} \, 1 - F\left(\frac{\E_n(\rho^{\otimes n})}{\Tr \E_n(\rho^{\otimes n})}, \omega^{\otimes \floor{rn}} \right) &\leq 1 - \liminf_{n\to\infty}\, \frac{1-\delta_n}{1 - \delta_n + \mu_n\delta_n} (1-\ve_n) \\
    &= 1 - \liminf_{n\to\infty} \left( 1 + \frac{\mu_n \delta_n}{1 - \delta_n} \right)^{-1} (1-\ve_n)\\
    &= 0
\end{aligned}\end{equation}
and probability of success
\begin{equation}\begin{aligned}
    \liminf_{n\to\infty} \, p_n &\geq 1 - \limsup_{n\to\infty} \, \delta_n( 1 - \mu_n)\\
    &= 1 - \delta\\
    &> 0.
\end{aligned}\end{equation}
Since $\chi$ was arbitrary, any rate below $\frac{D^\infty_\FF(\rho)}{D^\infty_\FF(\omega)}$ is thus achievable, and the result follows.
\end{proof}

\begin{boxed}
\begin{lemma}[Probabilistic monotonicity of max-relative entropy]
\label{lem:dmax} For any probabilistic operation $\E \in \OO_{\RNG,\delta}$ and any $\eta < \Tr \E(\rho)$, it holds that
 \begin{equation}\begin{aligned}
     D^\eta_{\max,\FF}(\rho) \geq  D^{\eta/\Tr \E(\rho)}_{\max,\FF}\left(\frac{\E(\rho)}{\Tr \E(\rho)}\right) + \log \left( \Tr \E(\rho) - \eta\right) - \log (1+\delta).
 \end{aligned}\end{equation}
 \end{lemma}
 \end{boxed}
\begin{proof}
Consider any state $\rho'$ with $\frac12 \norm{\rho - \rho'}{1} \leq \eta$.
The first step is to leverage the strong monotonicity of the generalised robustness~\cite{regula_2018}. Explicitly, write $2^{D_{\max,\FF}(\rho')} = \inf \lset \lambda\geq 1 \bar \rho' \leq \lambda \sigma,\; \sigma \in \FF \rset$ and let $\sigma \in \FF$ be any feasible solution to this problem. Then
\begin{equation}\begin{aligned}
    \frac{\E(\rho')}{\Tr \E(\rho')} &\leq \lambda \,\frac{\E(\sigma)}{\Tr \E(\rho')}\\
    &= \lambda \, \frac{\Tr \E(\sigma)}{\Tr \E(\rho')}\,\frac{\E(\sigma)}{\Tr \E(\sigma)}\\
    &\leq \lambda \,(1+\delta) \, \frac{\Tr \E(\sigma)}{\Tr \E(\rho')}\,\sigma'\\
    &\leq \lambda \,(1+\delta) \, \frac{1}{\Tr \E(\rho')}\,\sigma',
\end{aligned}\end{equation}
where we used the fact that $\E$ is a positive and trace--non-increasing map, as well as that $\frac{\E(\sigma)}{\Tr \E(\sigma)} \leq (1+\delta) \sigma'$ for some $\sigma' \in \FF$ by definition of $\OO_{\RNG,\delta}$. Optimising over all feasible $\lambda$ gives
\begin{equation}\begin{aligned}
    D_{\max} \left( \frac{\E(\rho')}{\Tr \E(\rho')} \right) \leq D_{\max,\FF}(\rho') - \log \Tr \E(\rho') + \log (1+\delta).
\end{aligned}\end{equation}
Now, since $\frac12 \norm{\rho - \rho'}{1} \leq \eta$, we have $\lvert\Tr\E(\rho) - \Tr\E(\rho')\rvert \leq \eta$, so we can bound $\Tr \E(\rho') \geq \Tr \E(\rho) - \eta$. Furthermore, for any two states and any trace non-increasing positive map $\E$ 
an application of the data processing inequality for the generalised trace distance~\cite[Proposition~3.8]{tomamichel_2016} shows that~\cite[proof of Proposition~6]{regula_2023}
\begin{equation}\begin{aligned}\label{eq:tracenorm}
  \frac12 \norm{\frac{\E(\rho)}{\Tr \E(\rho)} - \frac{\E(\rho')}{\Tr \E(\rho')} }{1} &\leq  \frac12 \norm{\frac{\E(\rho)}{\Tr \E(\rho)} - \frac{\E(\rho')}{\Tr \E(\rho)} }{1} +  \frac12 \norm{\frac{\E(\rho')}{\Tr \E(\rho)} - \frac{\E(\rho')}{\Tr \E(\rho')} }{1}\\
  &= \frac12 \norm{ \frac{ \E(\rho) - \E(\rho')}{\Tr \E(\rho)}}{1} + \frac12 \norm{ \frac{\E(\rho') \Tr\E(\rho') - \E(\rho') \Tr \E(\rho)}{\Tr \E(\rho) \Tr \E(\rho')}}{1}\\
  &= \frac12 \norm{ \frac{ \E(\rho) - \E(\rho')}{\Tr \E(\rho)}}{1} + \frac12 \left| \frac{\Tr \E(\rho') - \Tr \E(\rho)}{\Tr \E(\rho)} \right|\\
  &\leq \frac{\frac12\norm{\rho-\rho'}{1}}{\Tr\E(\rho)},
\end{aligned}\end{equation}
which altogether gives
\begin{equation}\begin{aligned}
    D^{\eta/\Tr\E(\rho)}_{\max,\FF} \left( \frac{\E(\rho)}{\Tr \E(\rho)} \right) \leq D_{\max,\FF} \left( \frac{\E(\rho')}{\Tr \E(\rho')} \right) \leq D_{\max,\FF}(\rho') - \log (\Tr \E(\rho)-\eta) + \log (1+\delta).
\end{aligned}\end{equation}
Optimising over all feasible $\rho'$ yields the statement of the Lemma.
\end{proof}

\section*{Entanglement theory}

Our results in the subsequent sections will specialise to the theory of entanglement. We use $\SEP$ to refer to the separable states, i.e.\ the free states of this resource theory. Resource--non-generating operations $\OO_\RNG$ correspond here to non-entangling operations NE, and asymptotically resource--non-generating operations $\OO_\ARNG$ to asymptotically non-entangling operations $\ANE$.

\section{Equivalence of strong converse and probabilistic distillation in entanglement theory}\label{sec:app_ent}

\let\e\xi

Recall that $E_{d,\OO}^{p=1} (\rho) \coloneqq r_{p=1}(\rho \to \Phi_+)$, $E_{d,\OO}^{p=1,\dagger} (\rho) \coloneqq r_{p=1}^\dagger(\rho \to \Phi_+)$, and $E_{d,\OO}^{p>0} (\rho) \coloneqq r_{p>0}(\rho \to \Phi_+)$.

In order to study transformations of entangled states under different free operations in entanglement theory, in particular LOCC, it is important to note that now it does not suffice to study sub-normalised quantum operations $\E$ like we did previously for RNG operations, as we must ensure that the overall instrument $(\E,\E')$ can be realised as a free protocol. To simplify the considerations, we make two observations:
\begin{enumerate}[(i)]
  \item Let $\OO$ be any class of operations that includes local operations and shared randomness (LOSR). Due to isotropic twirling~\cite{horodecki_1999}, the output of any distillation protocol with $\OO$ can be assumed without loss of generality to be of the form
  \begin{equation}\begin{aligned}\label{eq:isotropic}
    \varsigma(m, \ve) \coloneqq \left(1- \ve \right) \Phi_2^{\otimes m} + \ve \tau_{2^m},
  \end{aligned}\end{equation}
  where
  \bb
\Phi_d \coloneqq \ketbra{\Phi_d}\, ,\qquad \ket{\Phi_d}\coloneqq \frac{1}{\sqrt{d}} \sum_{i=0}^{d-1} \ket{ii}\, ,\qquad \tau_d \coloneqq \frac{\id - \Phi_d}{d^2-1}\,.
\ee

Now, given a positive integer $m\in \NN$ representing the number of local qubits and two error probabilities $\e,\delta\in [0,1]$ such that $\e+\delta\leq 1$, let us define the $(2^m+1)\times (2^m+1)$ bipartite state
\bb
\omega(m,\e, \delta) \coloneqq&\ \left(1- \e - \delta \right) \Phi_2^{\otimes m} + \e\, \tau_{2^m} + \delta \ketbra{ee} \\
=&\ \left( 1-\frac{\e}{1-4^{-m}} - \delta \right) \Phi_2^{\otimes m} + \frac{\e}{1-4^{-m}} \left(\frac{\id}{4}\right)^{\otimes m} + \delta \ketbra{ee} ,
\label{omega_m_eps_delta}
\ee
where $\ket{e}$ represents a local error flag that is orthogonal to the Hilbert space corresponding to the $m$ qubits.

 \item Instead of probabilistic transformations into isotropic states as in~\eqref{eq:isotropic}, we can then consider deterministic transformations into $\omega(m,\e,q)$. This is because, on the one hand, for any probabilistic distillation protocol which results in the state $\varsigma(m, \ve)$ with some probability $p$, if the protocol fails, Alice and Bob can simply prepare the state $\proj{ee}$ with probability $1-p$, thus obtaining the state $\omega(m,\ve p,1-p)$ deterministically. On the other hand, any deterministic protocol which results in the state $\omega(m,\e,\delta)$ can be easily modified to yield the state $\varsigma(m, \e/(1-\delta))$ with probability $1-\delta$ by simply 
measuring whether an error occurred or not. In particular, for any non-vanishing $p$, the error $\ve$ vanishes in the limit $m \to \infty$ if and only if so does $\e$.

\end{enumerate}

Let $\OO$ be a class of free operations in entanglement theory that includes LOCC. Given as above $\e,\delta\in [0,1)$ such that $\e+\delta<1$, let us define the corresponding one-shot probabilistically distillable entanglement by
\bb
E_{d,\OO}^{(1), \e, \delta}(\rho) \coloneqq \max \lsetr  m\in \NN \barr \exists\ \E\in \OO \cap {\rm CPTP}:\ \E(\rho) = \omega(m,\e,\delta) \rsetr ,
\ee
where $\rho$ is a generic state, and $\omega(m,\e,\delta)$ is defined by~\eqref{omega_m_eps_delta}. Asymptotically, we can set
\bb
E_{d,\OO}^{\e, \delta}(\rho) \coloneqq \liminf_{n\to\infty} \frac1n\, E_{d,\OO}^{(1), \e, \delta}(\rho^{\otimes n}).
\ee

\begin{boxed}
\begin{lemma} 
Let $m\in \NN$ be a positive integer, and consider $\e,\delta\in [0,1]$ with $\e+\delta\leq 1$. Then, for all $\lambda\in [0,1]$ and all $k\in \NN$ with $k<m$ the following transitions are possible via LOCC:
\begin{alignat}{4}
&\omega (m,\e, \delta) \to \omega(m'\!,\e'\!, \delta')\, ,\qquad && m'= m\, ,\quad && \e'= \e + \lambda \delta (1-2^{-m})\, ,\quad && \delta' = \delta (1-\lambda) \, ; \label{transformation_type_1} \\
&\omega (m,\e, \delta) \to \omega(m''\!\!,\e''\!\!, \delta'')\, ,\qquad && m''\hspace{-2.2pt} = m-k\, ,\quad &&\e'' \hspace{-2.2pt} = \frac{1-4^{-m+k}}{1-4^{-m}}\,\frac{\e}{2^k}\, ,\quad &&\delta''\hspace{-2.2pt} = \delta + \frac{1-2^{-k}}{1-4^{-m}}\, \e\, . \label{transformation_type_2}
\end{alignat}
\end{lemma}
\end{boxed}

\begin{proof}
The transformation of type~\eqref{transformation_type_1} is obtained by the following procedure:
\begin{enumerate}[(i)]
\item Alice and Bob check if the system is in the error state $\ket{e}$; if it is not, they do nothing;
\item if it is, with some probability $\lambda$ they prepare the separable~\cite{horodecki_1999} state $2^{-m} \Phi_{2}^{\otimes m} + (1-2^{-m}) \tau_{2^m}$; with probability $1-\lambda$, they do nothing.
\end{enumerate}
The transformation of type~\eqref{transformation_type_2}, instead, is obtained by a different procedure. Intuitively, we sacrifice $k$ pairs of qubits to try and determine whether the system is in the maximally entangled state or nor. More precisely, we run the following protocol:
\begin{enumerate}[(i)]
\item Alice and Bob check if the system is in the error state $\ket{e}$; if it is, they do nothing;
\item if it is not, then they measure the first $r$ qubits in the computational basis;
\item if upon communicating the outcomes they find that some of them differ, they declare an error, re-preparing the rest of the system in the state $\ket{ee}$;
\item if instead all outcomes are found to coincide pairwise, they do nothing on the remaining $m-k$ pairs of qubits.
\end{enumerate}
The elementary computations needed to verify that these protocols yield the claimed transformations are left to the reader.
\end{proof}

\begin{boxed}
\begin{proposition}\label{prop:dist}
Let $\OO$ be a class of free operations in entanglement theory that is closed under composition with LOCCs. Then, for a fixed $\rho$, the distillable entanglement $E_{d,\OO}^{\e,\delta}(\rho)$ depends on $\e,\delta\in [0,1)$ 
only through the sum $\e+\delta$, or equivalently through $p (1-\ve)$.
\end{proposition}
\end{boxed}

\begin{proof}
For asymptotically large $m$ and fixed $k$, transformations~\eqref{transformation_type_1}--\eqref{transformation_type_2} simplify considerably, becoming
\begin{alignat}{4}
&\omega (m,\e, \delta) \to \omega(m'\!,\e'\!, \delta')\, ,\qquad && m'= m\, ,\quad && \e' \approx \e + \lambda \delta\, ,\quad && \delta' = \delta (1-\lambda) \, ; \label{approx_type_1} \\
&\omega (m,\e, \delta) \to \omega(m''\!\!,\e''\!\!, \delta'')\, ,\qquad && m''\hspace{-2.2pt} = m-k\, ,\quad &&\e'' \hspace{-2.2pt} \approx \frac{\e}{2^k}\, ,\quad &&\delta''\hspace{-2.2pt} \approx \delta + (1-2^{-k})\, \e\, . \label{approx_type_2}
\end{alignat}
With these transformations, it is clear that we can turn a pair $(\e,\delta)$ with $\e+\delta<1$ into another pair $(\e',\delta')$ with $\e'+\delta'<1$ if and only if $\e+\delta = \e'+ \delta'$.
\end{proof}

A corollary of the above is Theorem~\ref{thm:entanglement} in the main text, whose extended statement we present here.
{
\renewcommand{\thetheorem}{3}
\begin{boxed}[filled]
\begin{theorem}
\label{thm:entanglement-s}
Let $\OO$ be any class of operations which is closed under composition with LOCC, i.e.\ such that $\E \in \OO,\, \F \in \rm LOCC \,\Rightarrow\, \F \circ \E \in \OO$. This includes in particular the set $\rm LOCC$ itself. 

Consider then any sequence of entanglement transformation protocols $(\E_n)_{n}$ with associated error $\limsup_{n\to\infty} \ve_n = \ve < 1$ and probability of success $\liminf_{n\to\infty}p_n = p > 0$. Then, then there exists another sequence of entanglement distillation protocols $(\E'_n)_n$ with error $\limsup_{n\to\infty} \ve'_n \eqqcolon \ve'$ and probability $\liminf_{n\to\infty} p'_n \eqqcolon p'$ if and only if
\begin{equation}\begin{aligned}
  p\, (1-\ve) = p' \, (1-\ve').
\end{aligned}\end{equation}
Hence, for all states $\rho$ and all sets of operations $\OO$ as above,
\begin{equation}\begin{aligned}\label{eq:thm_dist1}
E_{d,\OO}^{p>0} (\rho) = E_{d,\OO}^{p=1,\dagger}(\rho) = E_{d,\OO}^{p>0,\dagger}(\rho).
\end{aligned}\end{equation}

In particular, for all non-entangling operations it holds that
\begin{equation}\begin{aligned}\label{eq:thm_dist2}
  D^\infty_{\rm SEP}(\rho) &= E_{d,{\rm NE}}^{p>0}(\rho) = E_{d,{\rm NE}}^{p=1,\dagger}(\rho) = E_{d,{\rm NE}}^{p=1,\ddagger}(\rho)\\
  &= E_{d,{\rm ANE}}^{p>0}(\rho) = E_{d,{\rm ANE}}^{p=1,\dagger}(\rho) = E_{d,{\rm ANE}}^{p=1,\ddagger}(\rho).
\end{aligned}\end{equation}
\end{theorem}
\end{boxed}
}
The last statement in the above can be seen from the result that $D^\infty_{\rm SEP}(\rho) = E^{p=1,\dagger}_{d,\rm NE}(\rho) = E^{p=1,\ddagger}_{d,\rm ANE}(\rho)$~\cite{brandao_2010} (cf.~\cite[Corollary~III.3]{brandao_2010-1}), where the latter term denotes the strong converse distillable entanglement defined using the slightly stronger notion of a strong converse rate found in Eq.~\eqref{eq:def_sc_stronger}.\footnote{Whether this equivalence between the two definitions of the strong converse holds also for other sets of free operations $\OO$ is an interesting open problem.}

A more direct proof of the fact that the rate $D^\infty_{\rm SEP}(\rho)$ can be achieved by probabilistic non-entangling operations (as opposed to \emph{asymptotically} non-entangling ones as in Theorem~\ref{thm:reversibility-s}) can be obtained by following the proof of Theorem~\ref{thm:reversibility-s} but choosing $\omega_n = \Phi_2^{\otimes \floor{rn}}$ and realising that we can then pick $\pi_n$ to be separable states~\cite{vidal_1999,harrow_2003}, so no resources need to be generated whatsoever.


\section{Irreversibility of entanglement theory under probabilistic non-entangling operations}\label{sec:app_irrev}

Although here we will focus on (asymptotically) non-entangling operations (A)NE, i.e.\ (A)RNG maps with the set of free states $\FF = \SEP$, the discussion below applies verbatim also to the case of $\FF = \PPT$, which corresponds to (asymptotically) PPT--non-generating operations. This in particular implies irreversibility under PPT operations~\cite{rains_2001} enhanced by asymptotic entanglement non-generation.

Recall that we defined asymptotically resource--non-generating operations (and hence also asymptotically non-entangling ones, ANE) using the generalised robustness $R_\FF^g$ (see Eq.~\eqref{eq:def_rate_arng}). 
Let us consider two amended variants of this definition: one defined using the \emph{standard robustness}~\cite{vidal_1999}
\begin{equation}\begin{aligned}
  R^s_\SEP(\rho) &\coloneqq  \inf \lsetr \lambda \in \RR_+ \barr \frac{\rho + \lambda \sigma}{1+\lambda} \in \SEP,\; \sigma \in \SEP \rsetr,
  \end{aligned}\end{equation}
  and one using the \emph{negativity}~\cite{vidal_2002}
  \begin{equation}\begin{aligned}
    N(\rho) &\coloneqq \frac12 \left(\norm{\rho^\Gamma}{1} - 1\right)
  \end{aligned}\end{equation}
  with $\Gamma$ denoting partial transpose. The classes of operations are
\begin{equation}\begin{aligned}
  \OO_{\NE,\delta,s} &\coloneqq \lsetr \E \in \CPTNI \barr R^s_\SEP\left(\frac{\E(\sigma)}{\Tr \E(\sigma)}\right) \leq \delta\; \forall \sigma \in \SEP \rsetr,\\
  \OO_{\NE,\delta,N} &\coloneqq \lsetr \E \in \CPTNI \barr N\left(\frac{\E(\sigma)}{\Tr \E(\sigma)}\right) \leq \delta\; \forall \sigma \in \SEP \rsetr.
\end{aligned}\end{equation}
  The corresponding asymptotic transformation rates are defined analogously as before,
\begin{equation}\begin{aligned}
    r_{p>0}(\rho \toANEs \omega) \coloneqq \sup_{(\E_n)_n}  \Bigg\{ \,r \;\Bigg|\; & \lim_{n\to\infty} F\!\left( \frac{\E_n(\rho^{\otimes n})}{\Tr \E_n(\rho^{\otimes n})} ,\, \omega^{\otimes \floor{rn}} \right) = 1,\\
    & \lim_{n\to\infty} \,\sup_{\sigma \in \SEP} \,R^s_{\SEP} \!\left(\frac{\E_n(\sigma)}{\Tr \E_n(\sigma)}\right)  = 0,\quad \liminf_{n\to\infty}\, \Tr \E_n(\rho^{\otimes n}) > 0 \Bigg\},\\
        r_{p>0}(\rho \toANEN \omega) \coloneqq \sup_{(\E_n)_n}  \Bigg\{ \,r \;\Bigg|\; & \lim_{n\to\infty} F\!\left( \frac{\E_n(\rho^{\otimes n})}{\Tr \E_n(\rho^{\otimes n})} ,\, \omega^{\otimes \floor{rn}} \right) = 1,\\
    & \lim_{n\to\infty} \,\sup_{\sigma \in \SEP} \,N \!\left(\frac{\E_n(\sigma)}{\Tr \E_n(\sigma)}\right)  = 0,\quad \liminf_{n\to\infty}\, \Tr \E_n(\rho^{\otimes n}) > 0 \Bigg\}.
\end{aligned}\end{equation}
Noting that $R^s_\SEP(\rho) \geq N(\rho)$~\cite{vidal_2002}, for any state it holds that
\begin{equation}\begin{aligned}
  r_{p>0}(\rho \toNE \omega) \leq r_{p>0}(\rho \toANEs \omega) \leq r_{p>0}(\rho \toANEN \omega).
\end{aligned}\end{equation}
That is, imposing the asymptotic vanishing of the negativity is the weakest of the constraints that we consider, and thus $\ANE,N$ form the most permissive type of transformations. In particular, any transformation 
that is irreversible under $\ANE,N$ is also irreversible under $\ANE,s$ and $\ANE$.

As before, we define the probabilistic distillable entanglement $E_{d,\ANEN}^{p>0}(\rho) \coloneqq r_{p>0}(\rho \toANEN \Phi_+)$ and $E_{d,\ANEs}^{p>0}(\rho) \coloneqq r_{p>0}(\rho \toANEs \Phi_+)$ as well as the probabilistic entanglement cost $E_{c,\ANEN}^{p>0}(\rho) \coloneqq r_{p>0}(\Phi_+ \toANEN \rho)^{-1}$ and $E_{c,\ANEs}^{p>0}(\rho) \coloneqq r_{p>0}(\Phi_+ \toANEs \rho)^{-1}$.

In~\cite{lami_2023}, it was shown that in the deterministic setting there exists a two-qutrit state $\omega_3$ such that
\begin{equation}\begin{aligned}
  E_{d,\ANEN}^{p=1} (\omega_3) < E_{c,\ANEN}^{p=1} (\omega_3).
\end{aligned}\end{equation}
Here we generalise this to the probabilistic setting.
{
\renewcommand{\thetheorem}{2}
\begin{boxed}[filled]
\begin{theorem}\label{thm:prob_irrev-s}
The state $\omega_3$ satisfies that
\begin{equation}\begin{aligned}
  E_{d,\ANEN}^{p>0} (\omega_3) = E_{d,\ANEN}^{p=1} (\omega_3) = E_{d,\NE}^{p=1} (\omega_3) < E_{c,\NE}^{p=1} (\omega_3) = E_{c,\ANEN}^{p=1} (\omega_3) = E_{c,\ANEN}^{p>0} (\omega_3).
\end{aligned}\end{equation}
In particular, even in the probabilistic setting, the theory of entanglement remains irreversible under non-entangling operations, or under operations that generate asymptotically vanishing amounts of entanglement as quantified by the standard robustness $R^s_\SEP$ or the negativity $N$.\\[-6pt]

The same remains true if we replace $\SEP$ with the set $\PPT$ and {\textnormal{(A)NE}} operations with (asymptotically) PPT--non-generating operations.
\end{theorem}
\end{boxed}
}

The proof of the theorem will rely on two lemmas that establish a characterisation of the probabilistic rates of distillation and dilution, respectively.

\begin{boxed}
\begin{lemma}[Relative entropy and equivalence of probabilistic distillation]\label{lem:irrev_dist}
For any state $\rho$, the probabilistic distillable entanglement under $\ANEs$ or $\ANEN$ operations equals that under $\ANE$, that is, it is given by the regularised relative entropy:
\begin{equation}\begin{aligned}
 E_{d,\ANEs}^{p>0} (\rho) =  E_{d,\ANEN}^{p>0} (\rho) = E_{d,\ANE}^{p>0} (\rho) = D_\SEP^\infty(\rho).
\end{aligned}\end{equation}
\end{lemma}
\end{boxed}
\begin{proof}
By the result of Theorem~\ref{thm:entanglement-s}, we have that
\begin{equation}\begin{aligned}
  E_{d,\ANEN}^{p>0} (\rho) \geq E_{d,\NE}^{p>0}(\rho) =  D_\SEP^\infty(\rho),
\end{aligned}\end{equation}
so it remains to show the opposite direction. This can already be seen from Lemma~\ref{lem:hierarchy} combined with a result of~\cite{lami_2023}: we have the chain of inequalities
\bb
E_{d,\ANEN}^{p>0} (\rho) \textleq{(i)} E_{d,\ANE,N}^{p=1,\dagger} (\rho) \texteq{(ii)} E_{d,\NE}^{p=1,\dagger} (\rho) \texteq{(iii)} D_\SEP^\infty(\rho)\, ,
\ee
where (i)~follows from Lemma~\ref{lem:hierarchy}, (ii)~from~\cite[Lemma~S17]{lami_2023}, and finally (iii)~is known from the works of Brand\~{a}o and Plenio~\cite{brandao_2010,brandao_2010-1} (see also Theorem~\ref{thm:entanglement-s}).

For completeness, we will give a more direct alternative argument. To this end, we would like to establish a converse result as in Theorem~\ref{thm:reversibility-s}; however, the probabilistic monotonicity of $R^g_\SEP$ (Lemma~\ref{lem:dmax}) no longer holds when the operations $\OO_{\NE,\delta,N}$ are considered. Our idea will be to show that an approximate version of Lemma~\ref{lem:dmax} can still be established in the special case of distillation, that is, conversion into the maximally entangled state. This will in fact directly lead to a strong converse bound.

Consider then any sequence $(\E_n)_n$ of operations $\E_n \in \OO_{\NE,\delta_n,N}$ such that 
\bb
1 - \Tr \left( \frac{ \E_n(\rho^{\otimes n}) }{ \Tr \E_n(\rho^{\otimes n}) }\, \Phi_+^{\otimes \floor{rn}} \right) = 1-F\left( \frac{\E_n(\rho^{\otimes n})}{\Tr \E_n(\rho^{\otimes n})},\, \Phi_+^{\otimes \floor{rn}}\right) \eqqcolon \ve_n ,
\ee
with $\liminf_{n\to\infty} \Tr \E_n(\rho^{\otimes n}) \coloneqq p > 0$. It will be crucial to notice that, for any separable state $\sigma$, we have
\begin{equation}\begin{aligned}\label{eq:crucial}
  \Tr \left( \Phi_+^{\otimes \floor{rn}} \E_n(\sigma) \right) &= \Tr \left[ \left(\Phi_+^{\otimes \floor{rn}}\right)^\Gamma \E_n(\sigma)^\Gamma\right]\\
  &\leq \Tr \norm{\left(\Phi_+^{\otimes \floor{rn}}\right)^\Gamma}{\infty} \norm{\E_n(\sigma)^\Gamma}{1}\\
  &\leq \frac{1}{2^{\floor{rn}}} (1+2\delta_n) \Tr \E_n(\sigma)\\
  &\leq \frac{1}{2^{\floor{rn}}} (1+2\delta_n),
\end{aligned}\end{equation}
where in the third line we used the fact that the eigenvalues of 
$\left(\Phi_+^{\otimes \floor{rn}}\right)^\Gamma$ --- which is proportional to the swap operator --- are $\pm 2^{-\floor{rn}}$. We further leveraged the fact 
that
\begin{equation}\begin{aligned}
 \frac12 \left( \norm{\frac{\E_n(\sigma)^\Gamma}{\Tr \E_n(\sigma)^\Gamma}}{1} - 1 \right) \leq \delta_n
\end{aligned}\end{equation}
by the definition of $\OO_{\NE,\delta_n,N}$, and the last line follows since $\E_n$ is trace non-increasing.

Consider now any error sequence $(\zeta_n)_n$ with $\lim_{n\to\infty} \zeta_n = 0$, so that $\zeta_n < p$ and hence also $\zeta_n < \Tr\E_n(\rho^{\otimes n})$ holds for all sufficiently large $n$.
Let $\rho'_n$ be any state such that $\frac12 \norm{\rho^{\otimes n} - \rho'_n}{1} \leq \zeta_n$, and let $\sigma_n \in \SEP$ be any state such that $\rho'_n \leq \lambda_n \sigma_n$ for some $\lambda_n$. Then
\begin{equation}\begin{aligned}\label{eq:prob_strong_conv_error}
  1-\ve_n &= \Tr \left( \frac{ \E_n(\rho^{\otimes \floor{rn}}) }{ \Tr \E_n(\rho^{\otimes \floor{rn}}) }\, \Phi_+^{\otimes \floor{rn}} \right) \\
  &= \Tr \left( \frac{ \E_n(\rho'_n) }{ \Tr \E_n(\rho'_n) }\, \Phi_+^{\otimes \floor{rn}}\right) + \Tr \left( \left[\frac{ \E_n(\rho^{\otimes \floor{rn}}) }{ \Tr \E_n(\rho^{\otimes \floor{rn}}) } - \frac{ \E_n(\rho'_n) }{ \Tr \E_n(\rho'_n) } \right] \Phi_+^{\otimes \floor{rn}}\right)\\
  &\leq \lambda_n \Tr \left( \frac{ \E_n(\sigma_n) }{ \Tr \E_n(\rho'_n) }\, \Phi_+^{\otimes \floor{rn}}\right) + \frac12 \norm{\frac{ \E_n(\rho^{\otimes \floor{rn}}) }{ \Tr \E_n(\rho^{\otimes \floor{rn}}) } - \frac{ \E_n(\rho'_n) }{ \Tr \E_n(\rho'_n) } }{1}\\
  &\textleq{(iv)} \lambda_n \frac{1}{2^{\floor{rn}}} (1+2\delta_n) \frac{1}{\Tr \E_n(\rho'_n)} + \frac{\zeta_n}{\Tr \E_n(\rho'_n)}\\
  &\leq \left( \lambda_n \frac{1}{2^{\floor{rn}}} (1+2\delta_n)  + \zeta_n \right) \frac{1}{\Tr \E_n(\rho^{\otimes n}) - \zeta_n}.
\end{aligned}\end{equation}
Here, 
(iv) follows by Eq.~\eqref{eq:crucial} together with the probabilistic data processing inequality for the trace distance that we previously showed in Eq.~\eqref{eq:tracenorm}. 

So far, the above derivation made no assumption about the rate $r$. Let us now assume that
\begin{equation}\begin{aligned}
  r &> D^\infty_\SEP(\rho)\\
  &= \inf_{(\zeta_n)_n} \lset \limsup_{n\to\infty} \frac1n D^{\zeta_n}_{\max,\SEP}(\rho^{\otimes n}) \bar \lim_{n\to\infty} \zeta_n = 0 \rset,
\end{aligned}\end{equation}
where the second line follows from the asymptotic equipartition property of Brand\~{a}o--Plenio--Datta (see~\cite[Proposition~IV.2]{brandao_2010} \cite[Theorem~1]{datta_2009-2}). This implies that there exists an asymptotically vanishing error sequence $(\zeta_n)_n$ such that
\begin{equation}\begin{aligned}
  \limsup_{n\to\infty} \frac1n D^{\zeta_n}_{\max,\SEP}(\rho^{\otimes n}) < r,
\end{aligned}\end{equation}
and hence that there exist states $\rho'_n$ with $\frac12 \norm{\rho^{\otimes n} - \rho'_n}{1} \leq \zeta_n$ and $\rho'_n \leq \lambda_n \sigma_n$ such that
\begin{equation}\begin{aligned}
  c \coloneqq r - \frac1n \log \lambda_n > 0
\end{aligned}\end{equation}
for all sufficiently large $n$. 
Plugging this into~\eqref{eq:prob_strong_conv_error}, we have that
\begin{equation}\begin{aligned}\label{eq:goes_to_zero}
 1-\ve_n &\leq  \left( 2^{\log \lambda_n} 2^{-rn + 1} (1+2\delta_n)  + \zeta_n \right) \frac{1}{\Tr \E_n(\rho^{\otimes n}) - \zeta_n}\\
 &= \left( 2^{- c n  + 1} (1+2\delta_n)  + \zeta_n \right) \frac{1}{\Tr \E_n(\rho^{\otimes n}) - \zeta_n}.
\end{aligned}\end{equation}
Recalling that $\lim_{n\to\infty} \zeta_n = 0$ while $\Tr \E_n(\rho^{\otimes n})$ is lower bounded by a constant, if we furthermore assume that $\delta_n$ vanishes in the limit $n\to\infty$ (or indeed even only that $\delta_n = 2^{o(n)}$), the whole term on the right-hand side of Eq.~\eqref{eq:goes_to_zero} goes to zero, and we thus obtain 
\begin{equation}\begin{aligned}
  \lim_{n\to\infty} \ve_n = 1.
\end{aligned}\end{equation}
Since this holds for any $r > D^\infty_\SEP(\rho)$, it means that no rate achievable with error less than one can be larger than $D^\infty_\SEP$. Therefore
\begin{equation}\begin{aligned}
  E_{d,\ANE,N}^{p>0}(\rho) \leq E_{d,\ANE,N}^{p>0,\ddagger}(\rho) \leq D^\infty_\SEP(\rho) , 
\end{aligned}\end{equation}
where $E_{d,\ANE,N}^{p>0,\ddagger}$ is constructed by using the definition of rate given in Eq.~\eqref{eq:def_prob_sc_stronger}.

\end{proof}

Our approach to the bounds on entanglement cost will be based on the ideas in~\cite{lami_2023}, where the monotone known as the (logarithmic) \emph{tempered negativity}
\begin{equation}\begin{aligned}
  E_\tau(\rho) \coloneqq \log \max \lset \Tr X \rho \bar \norm{X^\Gamma}{\infty} \leq 1,\; \norm{X}{\infty} \leq \Tr X \rho \rset
\end{aligned}\end{equation}
was introduced.

\begin{boxed}
\begin{lemma}[Restrictions on probabilistic entanglement cost]\label{lem:irrev_cost}
For any state $\rho$, the probabilistic entanglement cost under $\ANEs$ and $\ANEN$ operations is lower bounded by the tempered negativity:
\begin{equation}\begin{aligned}
  E_{c,\ANEs}^{p>0} (\rho) \geq E_{c,\ANEN}^{p>0} (\rho) \geq E_{\tau}(\rho).
\end{aligned}\end{equation}
For the operations $\ANEs$, we can make an even stronger statement: the probabilistic and deterministic entanglement cost are equal, and in fact they both equal the entanglement cost under strictly non-entangling operations:
\begin{equation}\begin{aligned}
  E_{c,\ANEs}^{p>0} (\rho) = E_{c,\ANEs}^{p=1}(\rho) = E_{c,\rm NE}^{p=1} (\rho).
\end{aligned}\end{equation}
\end{lemma}
\end{boxed}
Both parts of the lemma rely on a probabilistic monotonicity result for $R^s_\SEP$, similar to the one that we previously showed for $R^g_\FF$ in Lemma~\ref{lem:dmax}. We state it as a separate lemma for clarity.

\begin{lemma}[Probabilistic monotonicity of standard robustness and negativity]
\label{lem:rs_monoton} 
 For any probabilistic operation $\E \in \OO_{\NE,\delta,N}$, it holds that
 \begin{equation}\begin{aligned}\label{eq:monoton_n}
     1 + R^s_\SEP(\rho) \geq  \norm{\frac{\E(\rho)^\Gamma}{\Tr \E(\rho)}}{1} \frac{ \Tr \E(\rho)}{ 2(1+2\delta) }.
 \end{aligned}\end{equation}

For any probabilistic operation $\E \in \OO_{\NE,\delta,s}$, it holds that
 \begin{equation}\begin{aligned}\label{eq:monoton_rs}
     1 + R^{s}_\SEP(\rho) \geq  \left[ 1 + R^{s}_{\SEP}\left(\frac{\E(\rho)}{\Tr \E(\rho)}\right)\right] \frac{ \Tr \E(\rho)}{ 1+2\delta }.
 \end{aligned}\end{equation}
 \end{lemma}
 \begin{proof}
For the case of $\E \in \OO_{\NE,\delta,N}$, consider any feasible solution for $R^s_\SEP(\rho)$, that is, a decomposition $\rho = (1+\lambda) \,\sigma_+ - \lambda \sigma_-$ with $\sigma_\pm \in \SEP$. Then
\begin{equation}\begin{aligned}
  \norm{ \frac{\E(\rho)^\Gamma}{\Tr \E(\rho)}}{1} &\leq (1+\lambda) \frac{\Tr \E(\sigma_+)}{\Tr \E(\rho)} \norm{\frac{\E(\sigma_+)^\Gamma}{\Tr \E(\sigma_+)}}{1} + \lambda \frac{\Tr \E(\sigma_-)}{\Tr \E(\rho)} \norm{\frac{\E(\sigma_-)^\Gamma}{\Tr \E(\sigma_-)}}{1}\\
  &\leq \left[ (1+\lambda) \frac{\Tr \E(\sigma_+)}{\Tr \E(\rho)} + \lambda \frac{\Tr \E(\sigma_-)}{\Tr \E(\rho)} \right] (1 + 2\delta)\\
  &\leq \frac{ 1+2 \lambda}{\Tr \E(\rho)}  (1 + 2\delta)\\
  &\leq 2 \frac{ 1+\lambda}{\Tr \E(\rho)}  (1 + 2\delta),
\end{aligned}\end{equation}
where in the second line we used that $N\!\left(\frac{\E(\sigma_\pm)}{\Tr \E(\sigma_\pm)}\right) \leq \delta$ by definition of $\OO_{\NE,\delta,N}$, and in the third line we used that the operation $\E$ may not increase trace. Optimising over all feasible $\lambda$ we get the stated result.

For $\OO_{\NE,\delta,s}$, we will use three equivalent definitions of $R^s_\SEP$, all of which follow straightforwardly from the original definition by using the fact that $\Tr \rho = 1$:
\begin{equation}\begin{aligned}
  R^{s}_\SEP(\rho) &= \min \lset \lambda \bar \rho = (1+\lambda) \,\sigma_+ - \lambda \sigma_-, \; \sigma_\pm \in \SEP \rset\\
  &= \min \lset \lambda \bar \rho \leq_\SEP (1+\lambda) \sigma_+,\; \sigma_+ \in \SEP  \rset\\
  &= \min \lset \lambda \bar \rho \geq_\SEP - \lambda \sigma_-, \;  \sigma_- \in \SEP \rset,
\end{aligned}\end{equation}
where $\leq_\SEP$ denotes inequality with respect to the cone of separable operators.
Let $\rho = (1+\lambda) \,\sigma_+ - \lambda \sigma_-$ be any feasible decomposition for $\rho$. 
Then, for any $\E \in \OO_{\NE,\delta,s}$, it holds that
\begin{equation}\begin{aligned}
  \frac{\E(\rho)}{\Tr \E(\rho)} &= (1+\lambda) \frac{\Tr \E(\sigma_+)}{\Tr \E(\rho)} \frac{\E(\sigma_+)}{\Tr \E(\sigma_+)} - \lambda \frac{\Tr \E(\sigma_-)}{\Tr \E(\rho)} \frac{\E(\sigma_-)}{\Tr \E(\sigma_-)}\\
  &\leq_\SEP (1+\lambda) \frac{\Tr \E(\sigma_+)}{\Tr \E(\rho)} (1+\delta) \,\sigma'_+ + \lambda \frac{\Tr \E(\sigma_-)}{\Tr \E(\rho)} \delta \, \sigma'_-
\end{aligned}\end{equation}
for some $\sigma'_\pm \in \SEP$, where we used that $\E$ can only generate at most $\delta$ robustness from any separable state. Since the operator in the last line is a 
non-negative combination of separable operators it is separable itself, and therefore it constitutes a feasible solution for the robustness of $\frac{\E(\rho)}{\Tr \E(\rho)}$. We can find its normalisation by simply taking the trace. Thus 
\begin{equation}\begin{aligned}
  1+ R^{s}_\SEP \left( \frac{\E(\rho)}{\Tr \E(\rho)} \right) &\leq \frac{ (1+\lambda) (1 + \delta) \, \Tr \E(\sigma_+) + \lambda \delta \, \Tr \E(\sigma_-)}{\Tr \E(\rho)} \\
  &\leq \frac{ (1+\lambda) (1+\delta) + \lambda \delta}{\Tr \E(\rho)} \\
  &\leq \frac{ (1+\lambda) (1+\delta) + (1+\lambda) \delta}{\Tr \E(\rho)} \\
  &= \frac{ (1+\lambda) (1 + 2 \delta)}{\Tr \E(\rho)}
\end{aligned}\end{equation}
using the trace--non-increasing property of $\E$.
 \end{proof}

 \begin{proof}[\bfseries\upshape Proof of Lemma~\ref{lem:irrev_cost}]
Consider first $\ANEN$ cost. Let $(\E_n)_n$ be any sequence of operations $\E_n \in \OO_{\NE,
\delta_n,N}$ such that $\frac12 \norm{\frac{\E_n(\Phi_+^{\otimes n})}{\Tr \E_n(\Phi_+^{\otimes n})} - \rho^{\otimes \floor{rn}}}{1} \eqqcolon \ve_n$. 
Using that $R^s_\SEP(\Phi_+^{\otimes n}) = 2^n - 1$~\cite{vidal_1999}, we have
\begin{equation}\begin{aligned}
  n &= \log \left( 1 + R^s_\SEP(\Phi_+^{\otimes n}) \right)\\
&\textgeq{(i)} \log \norm{\frac{\E_n(\Phi_+^{\otimes n})^\Gamma}{\Tr \E_n(\Phi_+^{\otimes n})}}{1} + \log \Tr \E_n(\Phi_+^{\otimes n}) - \log 2 - \log(1+2\delta_n)\\
&\textgeq{(ii)} E_\tau\!\left(\rho^{\otimes \floor{rn}}\right) + \log (1 - 2 \ve_n)  + \log \Tr \E_n(\Phi_+^{\otimes n}) - \log 2 - \log(1+2\delta_n)\\
&\textgeq{(iii)} \floor{rn} E_\tau(\rho) + \log (1 - 2 \ve_n)  + \log \Tr \E_n(\Phi_+^{\otimes n}) - \log 2 - \log(1+2\delta_n),
\end{aligned}\end{equation}
where: (i) follows from Lemma~\ref{lem:rs_monoton}; (ii) follows from Proposition~S5 and Lemma~S6 of~\cite{lami_2023} (cf.~proofs of Theorems~1 and S16 therein); (iii) follows from the super-additivity of $E_\tau$, shown in Proposition~S5 of~\cite{lami_2023}.
Assuming that $\Tr \E_n(\Phi_+^{\otimes n}) = 2^{-o(n)}$, $\delta_n = 2^{o(n)}$ and $\limsup_n \ve_n < \frac{1}{2}$ (which are all, in fact, weaker than our assumptions) we can divide by $n$ and take the $\liminf$ as $n\to\infty$ to get
\begin{equation}\begin{aligned}
  r^{-1} \geq E_\tau(\rho),
\end{aligned}\end{equation}
which is what was to be shown.

For the case of $\ANEs$, it is immediate from the definitions that $E_{c,\ANEs}^{p>0} (\rho) \leq E_{c,\NE}^{p=1}(\rho)$, so we need to show the opposite inequality.
We know from the results of Brand\~ao and Plenio~\cite[Sec.~V]{brandao_2010} that
\begin{equation}\begin{aligned}\label{eq:bp_cost}
  E_{c,\NE}^{p=1}(\rho) = \inf_{(\zeta_n)_n} \lset \limsup_{n\to\infty} \frac1n \log\left[ 1 + R^{s,\zeta_n}_{\SEP}\!\left(\rho^{\otimes n}\right) \right] \bar \lim_{n\to\infty} \zeta_n = 0 \rset,
\end{aligned}\end{equation}
where we denoted
\begin{equation}\begin{aligned}
R^{s,\zeta}_\SEP(\rho) \coloneqq \min_{\frac12\norm{\rho' - \rho}{1} \leq \zeta} R^s_\SEP(\rho').
\end{aligned}\end{equation}
Although an exact expression for the quantity on the right-hand side of~\eqref{eq:bp_cost} is not known (unlike in the asymptotic equipartition property for $R^g_\SEP$ that we used earlier), we can still use it to establish a converse bound. Take any sequence of operations $(\E_n)_n$ with $\E_n \in \OO_{\NE,\delta_n,s}$ such that $\frac12 \norm{\frac{\E_n(\Phi_+^{\otimes n})}{\Tr \E_n(\Phi_+^{\otimes n})} - \rho^{\otimes \floor{rn}}}{1} \eqqcolon \ve_n$ with $\liminf_n \Tr \E_n(\Phi_+^{\otimes n}) > 0$ and $\lim_n \ve_n = 0$. We then have that
\begin{equation}\begin{aligned}
    n &= \log \left( 1 + R^s_\SEP(\Phi_+^{\otimes n}) \right)\\
    &\geq \log\left( 1 + R^{s}_{\SEP}\!\left(\frac{\E_n(\Phi_+^{\otimes n})}{\Tr \E_n(\Phi_+^{\otimes n})}\right) \right) + \log \Tr \E_n(\Phi_+^{\otimes n}) - \log(1+2\delta_n)\\
    &\geq \log\left(  1 + R^{s,\,\ve_n}_{\SEP}\!\left(\rho^{\otimes \floor{rn}}\right) \right) + \log \Tr \E_n(\Phi_+^{\otimes n}) - \log(1+2\delta_n)
\end{aligned}\end{equation}
using Lemma~\ref{lem:rs_monoton}. As before, the two rightmost terms will vanish asymptotically provided that $\Tr \E_n(\Phi_+^{\otimes n}) = 2^{-o(n)}$ and $\delta_n = 2^{o(n)}$. Dividing by $n$, taking the $\limsup$, and using that $\lim_{n\to\infty} \ve_n = 0$  gives
\begin{equation}\begin{aligned}
  1 &\geq \inf_{(\zeta_n)_n} \lset \limsup_{n\to\infty} \frac1n \log\left[ 1 + R^{s,\,\zeta_n}_{\SEP}\!\left(\rho^{\otimes \floor{rn}}\right) \right] \bar \lim_{n\to\infty} \zeta_n = 0\rset.
\end{aligned}\end{equation}
To argue that $r^{-1} \geq E_{c,\NE}^{p=1}(\rho)$, it thus remains to show that
\begin{equation}\begin{aligned}\label{eq:rs_additivity}
  &\inf_{(\zeta_n)_n} \lset \limsup_{n\to\infty} \frac1n \log\left[ 1 + R^{s,\,\zeta_n}_{\SEP}\!\left(\rho^{\otimes \floor{rn}}\right) \right] \bar \lim_{n\to\infty} \zeta_n = 0\rset \\
  &\geq r \, \inf_{(\zeta_n)_n} \lset \limsup_{n\to\infty} \frac1n \log\left[ 1 + R^{s,\,\zeta_n}_{\SEP}\!\left(\rho^{\otimes n}\right) \right] \bar \lim_{n\to\infty} \zeta_n = 0\rset\\
  &= r \, E_{c,\NE}^{p=1}(\rho).
\end{aligned}\end{equation}
To see this, let us begin by defining for any $n \in \NN$ the corresponding quantity
\begin{equation}\begin{aligned}
  n^{(r)} \coloneqq \min \lset \floor{r l} \bar l \in \NN,\; \floor{rl} \geq n \rset.
\end{aligned}\end{equation}
For any fixed asymptotically vanishing error sequence $(\zeta_n)_n$, we also define
\begin{equation}\begin{aligned}\label{eq:aaaa}
  \zeta_n^{(r)} \coloneqq \max \lset \zeta_l \bar l \in \NN,\; \floor{rl} \geq n \rset.
\end{aligned}\end{equation}
Then
\begin{equation}\begin{aligned}
  \limsup_{n\to\infty} \frac1n \log\left[ 1 + R^{s,\,\zeta_n}_{\SEP}\!\left(\rho^{\otimes \floor{rn}}\right) \right] &= \limsup_{n\to\infty} \frac{\floor{rn}}{n} \frac{1}{\floor{rn}} \log\left[ 1 + R^{s,\,\zeta_n}_{\SEP}\!\left(\rho^{\otimes \floor{rn}}\right) \right]\\
  &= r\, \limsup_{n\to\infty} \frac{1}{\floor{rn}} \log\left[ 1 + R^{s,\,\zeta_n}_{\SEP}\!\left(\rho^{\otimes \floor{rn}}\right) \right]\\
  &\textgeq{(i)} r\, \limsup_{n\to\infty} \frac{1}{n^{(r)}} \log\left[ 1 + R^{s,\,\zeta_n^{(r)}}_{\SEP}\!\left(\rho^{\otimes n^{(r)}}\right) \right]\\
  &\textgeq{(ii)} r\, \limsup_{n\to\infty} \frac{1}{n^{(r)}} \log\left[ 1 + R^{s,\,\zeta_n^{(r)}}_{\SEP}\!\left(\rho^{\otimes n} \right) \right]\\
  &\textgeq{(iii)} r\, \limsup_{n\to\infty} \frac{1}{n+\ceil{r}} \log\left[ 1 + R^{s,\,\zeta_n^{(r)}}_{\SEP}\!\left(\rho^{\otimes n} \right) \right]\\
   &\textgeq{(iv)} r\,  \inf_{(\zeta'_n)_n} \lset \limsup_{n\to\infty} \frac1n \log\left[ 1 + R^{s,\,\zeta'_n}_{\SEP}\!\left(\rho^{\otimes n}\right) \right] \bar \lim_{n\to\infty} \zeta'_n = 0\rset,
\end{aligned}\end{equation}
where: (i) follows since $(n^{(r)})_n$ is simply a subsequence of $(\floor{rn})_n$, up to possible repetitions, and for any $l$ such that $n^{(r)} = \floor{rl}$, it holds that $\zeta_n^{(r)} \geq \zeta_l$ by definition; (ii) $n^{(r)} \geq n$ and discarding copies can only decrease $R^s_\SEP$ (as per Axiom~III, clearly obeyed by the resource theory of entanglement); (iii) is a consequence of the fact that
\begin{equation}\begin{aligned}
  n^{(r)} - \ceil{r} = \floor{rl} - \ceil{r} \leq \floor{rl - r} = \floor{r (l-1)} \leq n
\end{aligned}\end{equation}
since each $n^{(r)}$ is the \emph{least} value of $\floor{rl}$ such that $\floor{rl} \geq n$; finally, (iv) follows since adding a constant in the denominator will not affect the asymptotic limit, and $(\zeta_n^{(r)})_n$ is still a vanishing error sequence. This shows Eq.~\eqref{eq:rs_additivity} and thus concludes the proof of the Lemma.
 \end{proof}

\begin{proof}[\upshape\bfseries Proof of Theorem~\ref{thm:prob_irrev-s}] 
In~\cite[Theorem~S9]{lami_2023} we showed that $E_{d,\NE}^{p=1}(\omega_3) = D^\infty_\SEP(\omega_3) = \log \frac32$. Lemma~\ref{lem:irrev_dist} then gives that $E_{d,\ANE,N}^{p>0}(\omega_3)$ equals the same value.

On the other hand, in~\cite[Theorem~S9]{lami_2023} we also showed that $E_{c,\NE}^{p=1}(\omega_3) = E_\tau(\omega_3) = 1$. By Lemma~\ref{lem:irrev_cost}, $E_{c,\ANE,N}^{p>0}(\omega_3)$ is the same and in particular is strictly larger than the distillable entanglement.
\end{proof}

We sincerely appreciate all readers who actually read all the way until this point and didn't just scroll down to see how long the paper is. 


\end{document}